\documentclass[journal]{IEEEtran}
\usepackage{tikz}
\usetikzlibrary{positioning,shapes,arrows,shapes.geometric}
\usepackage{amsmath,amssymb,amsthm,amsfonts,bbm,dsfont,stmaryrd}
\usepackage{slashbox,pict2e}
\usepackage{xcolor,color}
\usepackage{float,stfloats}
\usepackage{comment}
\usepackage{graphicx,epstopdf,psfrag, pifont,cite,multirow,rotating}
\usepackage{lineno,hyperref}
\modulolinenumbers[5]
\usepackage{setspace}
\usepackage{pgfplots}
\usepackage{enumitem}
\usepackage{lipsum}
\usepackage{varwidth}
\usepackage{yfonts}
\usepackage{tabularx,booktabs}
\usepackage[ruled,vlined,linesnumbered]{algorithm2e}

\newcommand{\RN}[1]{\textup{\uppercase\expandafter{\romannumeral#1}}}

\ifCLASSOPTIONcompsoc
\usepackage[caption=false, font=normalsize, labelfont=sf, textfont=sf]{subfig}
\else
\usepackage[caption=false, font=footnotesize]{subfig}
\fi

\newtheorem{Proposition}{Proposition}

\begin{document}

\setlength{\textfloatsep}{8pt}
\setlength{\abovecaptionskip}{2pt}
\setlength{\belowcaptionskip}{0pt}

\title{Massive MIMO-OFDM ISAC for Sparse ISAR Imaging: Joint Power and Subcarrier Allocation}

\author{Hamid~Reza~Hashempour, Yanjiao Li, \textit{Member, IEEE}, Jie Zhang, \textit{Senior Member, IEEE},\\
Hyundong Shin, \textit{Fellow, IEEE}, and Hien~Quoc~Ngo, \textit{Fellow, IEEE}
\thanks{H.~R.~Hashempour, J.~Zhang, and H.~Q.~Ngo are with the Centre for Wireless Innovation (CWI), Queen's University Belfast, BT3 9DT Belfast, U.K. (e-mail: \{h.hashempoor, jie.zhang, hien.ngo\}@qub.ac.uk). H. Q. Ngo is also with the Department of Electronic Engineering, Kyung Hee University, Yongin-si, Gyeonggi-do 17104, Republic of Korea.}
\thanks{Y.~Li is with the Institute of Engineering Technology, University of Science and Technology Beijing, 100083 Beijing, China (e-mail: yanjiaoli@ustb.edu.cn).}
\thanks{H.~Shin is with the Department of Electronics and Information Convergence Engineering, Kyung Hee University, Yongin-si, Gyeonggi-do 17104, Republic of Korea (e-mail: hshin@khu.ac.kr).}
}

\markboth{}%
{Hashempour \MakeLowercase{\textit{et al.}}:
Massive MIMO-OFDM ISAC for Sparse ISAR Imaging}

\maketitle

\begin{abstract}
This paper investigates a massive multiple-input multiple-output
(mMIMO) orthogonal frequency-division multiplexing (OFDM) framework
for integrated sensing and communication (ISAC) with inverse synthetic
aperture radar (ISAR) imaging. A dual-function architecture combines
communication precoding and dedicated sensing beamforming to support
simultaneous downlink communication and ISAR imaging. To address
incomplete measurements arising from intermittent pilot transmission
and sparse sensing-subcarrier activation, an adaptive reweighted
two-dimensional alternating direction method of multipliers (ADMM)
algorithm is developed for sparse ISAR reconstruction.
A joint resource-allocation framework is formulated to optimize
communication-subcarrier assignment, sensing-subcarrier selection,
and transmit power allocation under communication quality-of-service (QoS)
and sensing constraints. Exploiting channel hardening, analytical
full-band benchmarks based on statistical channel state information
(CSI) are derived for maximum-ratio (MR) and zero-forcing (ZF)
precoding, while a soft actor--critic (SAC) algorithm is developed
for sparse-sensing resource allocation.
Numerical results demonstrate that the adaptive ADMM algorithm
improves ISAR reconstruction accuracy over conventional methods,
while the SAC-based design achieves higher sum spectral efficiency
than both the analytical full-band ISAC benchmarks and the
deep deterministic policy gradient (DDPG) baseline under both MR
and ZF precoding, while satisfying the communication QoS and
sensing constraints.
\end{abstract}

\begin{IEEEkeywords}
Integrated sensing and communications (ISAC), inverse synthetic aperture radar (ISAR), massive MIMO, joint waveform design, beamforming, deep reinforcement learning.
\end{IEEEkeywords}

\section{Introduction}\label{intro}

\IEEEPARstart{I}{ntegrated} sensing and communication (ISAC) has emerged as a key enabling technology for sixth-generation (6G) wireless networks, enabling simultaneous communication and environmental sensing through the shared utilization of spectrum, hardware, and signal-processing resources \cite{Liu2022ISAC,Giordani2020,Sturm2011}. By integrating sensing and communication functionalities into a common platform, ISAC facilitates emerging applications such as autonomous transportation, digital twins, extended reality, and the Internet of Things (IoT), while improving spectral and energy efficiency \cite{Zhang2021Overview,Liu2023SeventyYears,Wu2024JointCommunicationsSensing}.

Massive multiple-input multiple-output (mMIMO) is widely recognized as a cornerstone technology for future wireless systems due to its large spatial degrees of freedom, high beamforming gains, and excellent spectral efficiency (SE) performance \cite{Bjornson2020Prospective,Marzetta2016,Lee2022EnergyEfficientMassiveMIMO}. The abundant antenna resources available in mMIMO systems make them particularly attractive for ISAC deployments, where communication and sensing functionalities can be simultaneously supported using a common transceiver architecture \cite{Wei2023ISACSurvey,Liao2024mMIMOISAC}. Consequently, mMIMO-ISAC has attracted significant research interest in beamforming design, power allocation, waveform optimization, and sensing-performance enhancement \cite{Du2025DLBasedISAC}.

Most existing ISAC research has focused on estimating target parameters such as range, velocity, and angle, which are sufficient for localization and tracking but provide limited information about the physical structure of extended targets \cite{Zhang2021Overview,Liu2022ISAC}. In contrast, inverse synthetic aperture radar (ISAR) enables high-resolution two-dimensional imaging by exploiting the relative motion between the sensing platform and the target, thereby facilitating target recognition and classification based on geometric and scattering characteristics \cite{Hashempour2017CSISAR,Hashempour2018,Hashempour2017DSP}. Owing to these advantages, ISAR has recently attracted increasing interest within ISAC systems.

Orthogonal frequency-division multiplexing (OFDM), a widely adopted waveform for ISAC, provides a flexible framework for joint communication and sensing \cite{Sturm2011,Niu2025InterferenceManagement,Dai2026MIMOOFDMISAC}. Recent studies have extended OFDM-based ISAC to ISAR imaging, including IEEE 802.11ad-based imaging frameworks and multiple signal classification (MUSIC)-assisted OFDM-ISAR sensing approaches \cite{Han2022ISAR,Zhang2025ISARMUSIC}. Meanwhile, sparse-aperture and sparse-reconstruction techniques have demonstrated considerable potential for improving ISAR imaging efficiency in standalone radar systems \cite{Zhang2012SparseISAR,Hashempour2020}. However, sparse ISAR imaging remains largely unexplored in OFDM-based ISAC systems. Existing OFDM-ISAR approaches typically rely on full-band sensing and fixed sensing resources, resulting in limited flexibility and reduced SE \cite{Hashempour2018,Han2022ISAR,Zhang2025ISARMUSIC}.

Meanwhile, extensive efforts have been devoted to resource allocation and beamforming design for mMIMO-ISAC systems. Existing studies have investigated transmit beamforming, power allocation, communication-rate maximization, beampattern optimization, and sensing-performance enhancement under various sensing metrics, including the signal-to-clutter-plus-noise ratio (SCNR), Cramér--Rao lower bound (CRLB), and mainlobe-to-average-sidelobe ratio (MASR) \cite{Liu2018MUMIMO,Liu2020JointBeamforming,Qi2022HybridBeamformingISAC,Liao2024mMIMOISAC,Nguyen2025MassiveMIMOISAC}. In particular, resource-allocation and beamforming designs for mMIMO-ISAC systems have been developed based on statistical channel state information (CSI), communication quality-of-service (QoS) requirements, and sensing-performance constraints \cite{Liao2024mMIMOISAC,Nguyen2025MassiveMIMOISAC}. 
Beyond conventional optimization-based and
game-theoretic approaches, deep reinforcement learning (DRL) has recently
attracted attention for dynamic ISAC resource
management. In particular, transformer-enabled
mean-field reinforcement learning has been
investigated for interference management in
ISAC-enabled space--air--ground integrated
networks \cite{SAGIN2026ISAC}, while generative
artificial intelligence (AI)-enabled DRL has been developed for
energy-efficient beamforming in simultaneously transmitting and reflecting reconfigurable intelligent surface (STAR-RIS)-aided
ISAC systems with hardware impairments
\cite{STARRIS2026ISAC}.
Although these works provide valuable insights into sensing-communication coexistence, they generally focus on point-target sensing and do not explicitly account for ISAR image reconstruction quality. Furthermore, the interaction between sparse sensing resource allocation and ISAR imaging performance remains largely unexplored.

Consequently, three important research gaps remain. First, sparse ISAR imaging has rarely been investigated within mMIMO-OFDM ISAC frameworks. Second, existing sparse ISAR reconstruction methods generally ignore the coupling between sensing-resource allocation and communication performance. Third, the tradeoff between ISAR reconstruction accuracy and communication SE has not been systematically characterized.
Motivated by these observations, this paper develops a unified mMIMO-OFDM ISAC framework that jointly addresses sparse ISAR imaging and sensing-communication resource allocation. The main contributions of this work are summarized as follows:
\begin{itemize}
\item We propose a mMIMO-OFDM ISAC framework for simultaneous downlink communication and ISAR imaging. A practical dual-function transmission architecture is adopted, where the transmitted signal is formed as a weighted combination of a communication precoder and a dedicated sensing beamformer.
\item We formulate a sparse-aperture ISAR sensing model arising from intermittent pilot transmission and sparse sensing-subcarrier activation. To recover high-resolution ISAR images from incomplete measurements, an adaptive reweighted two-dimensional alternating direction method of multipliers (ADMM) reconstruction algorithm is developed.
\item Leveraging the channel hardening property of mMIMO systems and statistical CSI, we derive tractable expressions for the communication signal-to-interference-plus-noise ratio (SINR) and sensing MASR. Based on these expressions, analytical benchmark resource-allocation schemes are developed for both maximum-ratio (MR) and zero-forcing (ZF) precoding.
\item We develop a DRL-based framework for joint resource allocation. Specifically, a soft actor--critic (SAC) algorithm is employed to jointly optimize communication-subcarrier assignment, sensing-subcarrier selection, and transmit power allocation subject to communication QoS and sensing constraints.
\item Numerical results demonstrate that the proposed adaptive
ADMM scheme improves sparse ISAR reconstruction accuracy over
conventional methods, while the SAC-based resource allocation
achieves higher sum SE than both the analytical full-band ISAC
benchmarks and the deep deterministic policy gradient (DDPG) baselines under both MR and ZF precoding.
The results further characterize the tradeoff between
communication SE and ISAR reconstruction accuracy while
satisfying the communication and sensing constraints.
\end{itemize}
The remainder of this paper is organized as follows. Section~\ref{sec:sys} presents the proposed mMIMO-OFDM ISAC system model. Section~\ref{sec:isar} develops the sparse ISAR imaging formulation and adaptive ADMM reconstruction algorithm. Section~\ref{sec:metrics} derives the communication and sensing performance metrics, formulates the joint resource-allocation problem, and develops an analytical full-band benchmark. Section~\ref{sec:DRL} presents the SAC-based framework for joint power control and subcarrier allocation. Numerical results are provided in Section~\ref{sec:sim}, followed by conclusions in Section~\ref{sec:conc}.

\textit{Notation}: Bold lowercase and uppercase letters denote vectors and matrices, respectively. The operators $(\cdot)^T$, $(\cdot)^H$, and $(\cdot)^*$ denote transpose, Hermitian transpose, and complex conjugate, respectively. The sets of real and complex matrices are denoted by $\mathbb{R}^{M\times N}$ and $\mathbb{C}^{M\times N}$. The matrix $\mathbf{I}_N$ denotes the $N\times N$ identity matrix. The operators $\operatorname{diag}(\cdot)$, $\operatorname{tr}(\cdot)$, and $\mathbb{E}(\cdot)$ denote the diagonalization, trace, and expectation operators, respectively. The Euclidean, Frobenius, and $\ell_1$ norms are denoted by $\Vert\cdot\Vert_2$, $\Vert\cdot\Vert_F$, and $\Vert\cdot\Vert_1$, respectively. The symbol $\otimes$ denotes the Kronecker product, and $\mathcal{CN}(\boldsymbol{\mu},\mathbf{C})$ denotes a circularly symmetric complex Gaussian distribution with mean $\boldsymbol{\mu}$ and covariance matrix $\mathbf{C}$.

\section{ISAC System Model}\label{sec:sys}

We consider a massive MIMO-enabled ISAC downlink system operating in the sub-6~GHz band, as illustrated in Fig.~\ref{fig_system}. The system adopts a co-located monostatic sensing architecture, in which the transmit and receive arrays are located at the same base station (BS). More precisely, the BS is equipped with $M_t$ transmit antennas and simultaneously serves $K$ single-antenna communication users (CUs) while illuminating an extended target scene for ISAR imaging. The reflected echoes are collected by a co-located sensing receiver equipped with $M_r$ receive antennas.

\begin{figure}[t]
    \centering
    \includegraphics[width=1\linewidth]{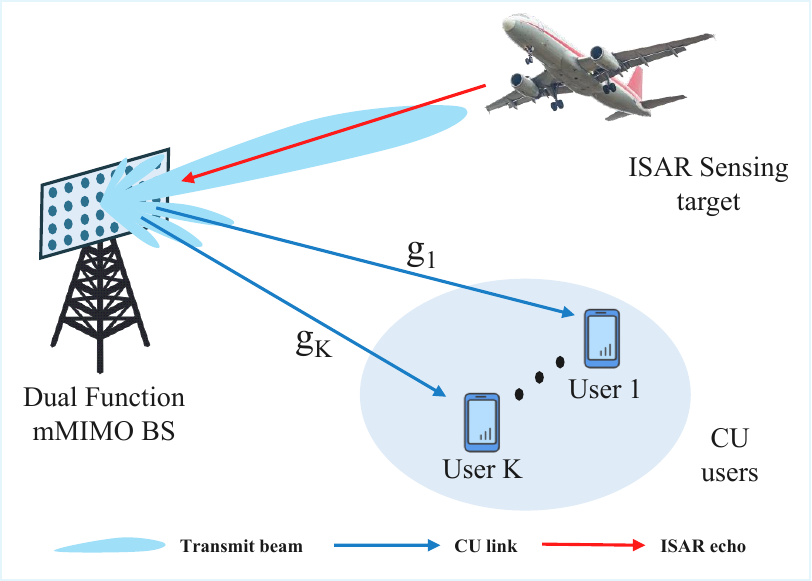}
    \caption{System model for OFDM-based ISAR-ISAC.}
    \label{fig_system}
\end{figure}

\subsection{Signal Model}
A cyclic-prefix OFDM (CP-OFDM) waveform is adopted for both communication and sensing. This choice is motivated by the fact that CP-OFDM is the baseline downlink waveform in 5G New Radio (NR) and has also been considered in recent 3GPP RAN WG1 discussions on 6G radio waveform design~\cite{R1_2506595_6G_waveform, R1_2506702_RAN1_122_report}. Hence, the same OFDM time-frequency grid can be exploited for communication data transmission and sensing echo acquisition, leading to a unified ISAC signal model. The system bandwidth is given by
\begin{equation}
B_{\mathrm{sys}} = N_c \Delta f,
\end{equation}
where $N_c$ is the number of subcarriers and $\Delta f$ is the subcarrier spacing. For high-resolution ISAR imaging, the dedicated sensing waveform is designed to exploit the full system bandwidth whenever full-band sensing is available \cite{Hashempour2018}. For example, in the considered setup, $B_{\mathrm{sys}}=100$~MHz. Full-band sensing improves the range resolution and avoids ambiguities caused by fragmented spectral occupation. In contrast, the CUs are assigned subcarriers according to their traffic demands and QoS requirements. Thus, a CU does not necessarily occupy the entire bandwidth.

Let $\rho_{k,n}\in\{0,1\}$ denote the subcarrier-allocation indicator for CU $k$ on subcarrier $n$, where $\rho_{k,n}=1$ indicates that subcarrier $n$ is assigned to CU $k$, and $\rho_{k,n}=0$ otherwise. Let $s_k[\ell,n]\in\mathbb{C}$ denote the communication data symbol intended for CU $k$ on OFDM symbol $\ell$ and subcarrier $n$, satisfying $\mathbb{E}\{|s_k[\ell,n]|^2\}=1$. Throughout this paper, $n\in\{0,\ldots,N_c-1\}$ indexes the OFDM subcarriers, $\ell\in\{0,\ldots,N_s-1\}$ indexes the OFDM symbols, and $N_s$ is the number of OFDM symbols. The communication symbols are stacked as
\begin{equation}
\mathbf{s}_{\ell,n}
=
\big[s_1[\ell,n],\ldots,s_K[\ell,n]\big]^T
\in \mathbb{C}^{K\times 1},
\end{equation}
and the corresponding subcarrier-allocation matrix is defined as
\begin{equation}
\mathbf{D}_n
=
\operatorname{diag}\{\rho_{1,n},\ldots,\rho_{K,n}\}.
\end{equation}

In addition to the communication data stream, a dedicated sensing stream with symbol $s_{\ell,n}^{\mathrm{sen}}\in\mathbb{C}$ is transmitted on subcarrier $n$, where $\mathbb{E}\{|s_{\ell,n}^{\mathrm{sen}}|^2\}=1$. The sensing symbols are assumed to be known at the sensing receiver and statistically independent of the communication symbols. To model both full-band ISAR sensing and sparse sensing measurements for compressed sensing, we introduce the sensing subcarrier-selection indicator $\rho_n^{\mathrm{sen}}\in\{0,1\}$, where $\rho_n^{\mathrm{sen}} =1$ implies subcarrier $n$ is allocated to sensing, and $\rho_n^{\mathrm{sen}} =0$ otherwise. The active sensing-subcarrier set is therefore defined as
\begin{equation}
\Omega_{\mathrm{sen}}
\triangleq
\left\{ n:\rho_n^{\mathrm{sen}}=1,\; n=0,\ldots,N_c-1 \right\}.
\end{equation}
The full-band sensing case corresponds to \(\rho_n^{\mathrm{sen}}=1\) for all $n$. When \(|\Omega_{\mathrm{sen}}|<N_c\), only a subset of subcarriers is used, yielding undersampled measurements. Since the ISAR scene is sparse in the range--Doppler domain \cite{Hashempour2020,Zhang2012SparseISAR}, it can be reconstructed using compressed sensing.

Let $\mathbf{t}_{\mathrm{c},k}[n]\in\mathbb{C}^{M_t\times 1}$ denote the communication beamforming vector for CU $k$ on subcarrier $n$. The communication precoding matrix on subcarrier $n$ is defined as
\begin{equation}
\mathbf{T}_{\mathrm{c}}[n]
\triangleq
\big[
\mathbf{t}_{\mathrm{c},1}[n],
\ldots,
\mathbf{t}_{\mathrm{c},K}[n]
\big]
\in\mathbb{C}^{M_t\times K}.
\end{equation}
Furthermore, let \(\mathbf w\in\mathbb C^{M_t\times 1}\) denote the unit-norm sensing beamforming vector, i.e., \( \Vert\mathbf w\Vert_2^2=1\), and let \(q_n\geq 0\) denote the sensing power allocated to subcarrier \(n\). Under the narrowband-array assumption, the same beamforming vector \(\mathbf w\) is used across all subcarriers, since the array response is approximately frequency-independent. Nevertheless, the sensing power may vary across subcarriers through \(\{q_n\}\).

For sensing, we consider a line-of-sight propagation scenario between the BS array and the extended target. The BS antenna array is modeled as a uniform planar array (UPA) deployed on the $y$--$z$ plane with half-wavelength antenna spacing. Let $M_y$ and $M_z$ denote the numbers of antenna elements along the $y$- and $z$-axes, respectively, such that $M_t=M_yM_z$. Let $\phi\in(-\pi/2,\pi/2)$ and $\theta\in(0,\pi)$ denote the azimuth and elevation angles of departure with respect to the $x$- and $z$-axes, respectively. The UPA steering vector is given by
\begin{equation}\label{trans_steer}
\mathbf{a}_{\mathrm{t}}(\phi,\theta)
=
\mathbf{a}_y(\phi,\theta)\otimes \mathbf{a}_z(\theta)
\in\mathbb{C}^{M_t\times 1},
\end{equation}
where $\otimes$ denotes the Kronecker product. The steering vectors along the $y$- and $z$-axes are respectively expressed as \cite{Liao2024mMIMOISAC}
\begin{align}
&\mathbf{a}_y(\phi,\theta)\nonumber\\
&=
\frac{1}{\sqrt{M_y}}
\left[
1,
e^{-j\pi \sin(\phi)\sin(\theta)},
\ldots,
e^{-j\pi (M_y-1)\sin(\phi)\sin(\theta)}
\right]^T,\nonumber\\
&\mathbf{a}_z(\theta)
=
\frac{1}{\sqrt{M_z}}
\left[
1,
e^{-j\pi \cos(\theta)},
\ldots,
e^{-j\pi (M_z-1)\cos(\theta)}
\right]^T.\nonumber
\end{align}
The sensing beamformer $\mathbf{w}$ is designed to focus energy toward the target direction while suppressing sidelobe radiation. The normalized transmit beampattern is given by $\left|\mathbf{a}_{\mathrm{t}}^{H}(\phi,\theta)\mathbf{w}\right|^2$ \cite{Liao2024mMIMOISAC}. Let $\alpha_{k,n}\geq 0$ denote the communication power allocated to CU $k$ on subcarrier $n$, and define
\begin{equation}
\mathbf{D}_{\boldsymbol{\alpha}_n}
\triangleq
\operatorname{diag}\{\alpha_{1,n},\ldots,\alpha_{K,n}\}.
\end{equation}
Then, the communication precoding matrix is written as
\begin{equation}
\mathbf{T}_{\mathrm{c}}[n]
=
\mathbf{U}_{\mathrm{c}}[n]\mathbf{D}_{\boldsymbol{\alpha}_n}^{1/2},
\label{eq:comm_precoder_power_normalized}
\end{equation}
where \(\mathbf{U}_{\mathrm{c}}[n]=[\mathbf{u}_{1,n},\ldots,\mathbf{u}_{K,n}]\in\mathbb{C}^{M_t\times K}\) contains the normalized precoding directions. The normalization is chosen such that $\mathbb{E}\{\|\mathbf{u}_{k,n}\|_2^2\}=1$ for $k=1,\ldots,K$, where the expectation is with respect to the small-scale fading when the precoder depends on the channel estimates. Hence, the $k$-th communication beamformer can be written as $\mathbf{t}_{\mathrm{c},k}[n]=\sqrt{\alpha_{k,n}}\mathbf{u}_{k,n}$, and satisfies $\mathbb{E}\!\left\{\left\|\mathbf{t}_{\mathrm{c},k}[n]\right\|_2^2\right\}=\alpha_{k,n}$.

The $M_t$-antenna frequency-domain transmit vector on OFDM symbol $\ell$ and subcarrier $n$ is then expressed as
\begin{equation}
\mathbf{x}_{\ell,n}
=
\mathbf{T}_{\mathrm{c}}[n]\mathbf{D}_n\mathbf{s}_{\ell,n}
+
\sqrt{q_n\rho_n^{\mathrm{sen}}}\,
\mathbf{w}s_{\ell,n}^{\mathrm{sen}} .
\label{eq:x_ln}
\end{equation}
The first term in \eqref{eq:x_ln} represents the multiuser communication signal. The second term represents the dedicated sensing waveform transmitted over the selected sensing subcarriers. The average transmit power on subcarrier $n$ is given by
\begin{align}
P_n
&=
\mathbb{E}\!\left\{
\left\|
\mathbf{x}_{\ell,n}
\right\|_2^2
\right\}
\nonumber\\
&=
\mathbb{E}\!\left\{
\operatorname{tr}
\left(
\mathbf{T}_{\mathrm{c}}[n]\mathbf{D}_n
\mathbf{T}_{\mathrm{c}}^H[n]
\right)
\right\}
+
q_n\rho_n^{\mathrm{sen}}
\nonumber\\
&=
\sum_{k=1}^{K}
\rho_{k,n}\alpha_{k,n}
+
q_n\rho_n^{\mathrm{sen}} .
\label{eq:subcarrier_power}
\end{align}
Accordingly, the total BS transmit-power constraint is written as
\begin{equation}
\sum_{n=0}^{N_c-1}
\left(
\sum_{k=1}^{K}
\rho_{k,n}\alpha_{k,n}
+
q_n\rho_n^{\mathrm{sen}}
\right)
\leq
P_{\max},
\label{eq:total_power_constraint}
\end{equation}
where $P_{\max}$ denotes the maximum transmit power at the BS. Finally, the set of CUs scheduled on subcarrier $n$ is defined as
\begin{equation}
\mathcal{U}_n
\triangleq
\{k:\rho_{k,n}=1\}.
\end{equation}
Unlike conventional OFDMA systems that impose strict per-subcarrier user exclusivity, the considered mMIMO system allows multiple CUs to share the same subcarrier through spatial multiplexing. Therefore, a CU may occupy multiple subcarriers, and a subcarrier may be simultaneously assigned to multiple CUs.
\vspace{-2mm}

\subsection{Channel Model and MR/ZF Precoding}\label{subsec:csi}

Let $\mathbf{g}_{k}[n]\in\mathbb{C}^{M_t\times 1}$ denote the downlink channel vector from the BS to CU $k$ on subcarrier $n$. We adopt a Rayleigh fading model given by
\begin{equation}
    \mathbf{g}_{k}[n]
    =
    \sqrt{\xi_k}\mathbf{h}_{k}[n],
    \label{eq:dl_ch}
\end{equation}
where $\mathbf{h}_{k}[n]\sim\mathcal{CN}(\mathbf{0},\mathbf{I}_{M_t})$ denotes the small-scale fading vector, and $\xi_k$ is the corresponding large-scale fading coefficient. Since large-scale fading is mainly determined by path loss and shadowing, it is assumed to be independent of the subcarrier index.

In this paper, we adopt a time-division duplexing (TDD) protocol in which uplink pilot training is conducted at the beginning of each coherence interval, followed by downlink communication and ISAR sensing~\cite{Liao2024mMIMOISAC}. Each CU transmits an orthogonal pilot sequence of length $\tau_p$, with $\tau_p\geq K$, using pilot transmit power $\rho_p$. The uplink noise variance at the BS is denoted by $\sigma_{\mathrm{ul}}^2$. Based on the received pilot signals and channel reciprocity, the BS obtains the minimum mean square error (MMSE) estimate $\widehat{\mathbf{g}}_{k}[n]$ of $\mathbf{g}_{k}[n]$ on each subcarrier. Therefore, the channel can be decomposed as
\begin{equation}
\mathbf{g}_{k}[n]
=
\widehat{\mathbf{g}}_{k}[n]
+
\mathbf{e}_k[n],
\label{eq:ch_decomp}
\end{equation}
where $\mathbf{e}_k[n]$ denotes the channel estimation error. From the property of the MMSE estimator and Rayleigh fading channels, $\widehat{\mathbf{g}}_{k}[n]$ and $\mathbf{e}_k[n]$ are independent. Moreover,
\begin{align}
\widehat{\mathbf{g}}_{k}[n]
&=
\sqrt{\eta_k}\bar{\mathbf{g}}_{k}[n],
\label{eq:g_hat}\\
\bar{\mathbf{g}}_{k}[n]
&\sim
\mathcal{CN}(\mathbf{0},\mathbf{I}_{M_t}),
\label{eq:gbar}\\
\mathbf{e}_k[n]
&\sim
\mathcal{CN}(\mathbf{0},\varepsilon_k\mathbf{I}_{M_t}),
\label{eq:eps}
\end{align}
where $\eta_k$ denotes the variance of the channel estimate for CU $k$, and $\xi_k=\eta_k+\varepsilon_k$. The MMSE estimation-error variance and the variance of the channel estimate are respectively given by
\begin{align}
\varepsilon_k
&=
\frac{\sigma_{\mathrm{ul}}^2\xi_k}
{\sigma_{\mathrm{ul}}^2+\tau_p\rho_p\xi_k},
\label{eq:eps_def}\\
\eta_k
&=
\xi_k-\varepsilon_k
=
\frac{\tau_p\rho_p\xi_k^2}
{\sigma_{\mathrm{ul}}^2+\tau_p\rho_p\xi_k}.
\label{eq:eta_def}
\end{align}
For later use, we define the normalized channel-estimation error as $\delta_k \triangleq \frac{\varepsilon_k}{\xi_k} \in[0,1)$, which yields $\frac{\eta_k}{\xi_k} = 1-\delta_k$, and $\frac{\varepsilon_k}{\eta_k} = \frac{\delta_k}{1-\delta_k}$. In the perfect-CSI case, $\delta_k=0$, and hence $\eta_k=\xi_k$ and $\varepsilon_k=0$.

We consider two widely used linear precoding schemes in mMIMO, namely MR and ZF precoding. For MR precoding, the communication beamforming vector is chosen as
\begin{equation}
\mathbf{t}_{\mathrm{c},k}^{\mathrm{MR}}[n]
=
\sqrt{\frac{\alpha_{k,n}}{M_t}}\,
\bar{\mathbf{g}}_k[n],
\qquad k=1,\ldots,K.
\label{eq:mr_precoder_def}
\end{equation}
Let $L_n\triangleq|\mathcal{U}_n|$ denote the number of CUs scheduled on subcarrier $n$, and assume $M_t>L_n$. The ZF precoder on subcarrier $n$ is constructed based on the estimated channels of the scheduled CUs as
\begin{equation}
\mathbf{T}_{\mathrm{c},\mathcal{U}_n}^{\mathrm{ZF}}[n]
=
\sqrt{M_t-L_n}\,
\bar{\mathbf{G}}_n
\left(\bar{\mathbf{G}}_n^H\bar{\mathbf{G}}_n\right)^{-1}
\mathbf{D}_{\boldsymbol{\alpha}_{\mathcal{U}_n}}^{1/2},
\label{eq:zf_precoder_def}
\end{equation}
where $\bar{\mathbf{G}}_n \triangleq [\bar{\mathbf{g}}_i[n]]_{i\in\mathcal{U}_n} \in\mathbb{C}^{M_t\times L_n}$, and $\mathbf{D}_{\boldsymbol{\alpha}_{\mathcal{U}_n}} \triangleq \operatorname{diag}\{\alpha_{i,n}:i\in\mathcal{U}_n\}$. The corresponding columns are embedded into the full $M_t\times K$ precoding matrix $\mathbf{T}_{\mathrm{c}}[n]$, while unscheduled users are excluded by $\mathbf{D}_n$.

\section{OFDM-Based ISAR Imaging and Sparse Reconstruction}\label{sec:isar}
\subsection{OFDM-Based ISAR Imaging}

To connect the ISAC transmit model in \eqref{eq:x_ln} with the OFDM-based ISAR imaging process, the continuous-time baseband signal radiated by the BS antenna array is written as
\begin{equation}
\mathbf{x}(t)
=
\frac{1}{\sqrt{N_c}}
\sum_{\ell=0}^{N_s-1}
\sum_{n=0}^{N_c-1}
\mathbf{x}_{\ell,n}
e^{j2\pi n\Delta f (t-\ell T_t-T_g)}
u\!\left(\frac{t-\ell T_t}{T_t}\right),
\label{eq:tx_signal_mimo}
\end{equation}
where $\mathbf{x}_{\ell,n}\in\mathbb{C}^{M_t\times 1}$ is defined in \eqref{eq:x_ln}. Moreover, $T_u=1/\Delta f$ is the useful OFDM-symbol duration, $T_g$ is the cyclic-prefix or guard-interval duration, and $T_t=T_u+T_g$ is the total OFDM-symbol duration. The rectangular pulse $u(\tau)$ is defined as
\begin{equation}
u(\tau)
=
\begin{cases}
1, & 0\leq \tau < 1,\\
0, & \text{otherwise}.
\end{cases}
\end{equation}
The phase reference $t-\ell T_t-T_g$ aligns the subcarrier phase with the useful OFDM interval after CP removal.

We consider an extended target composed of $Q$ dominant scattering centers. The $i$th scattering center has local coordinates $(x_i,y_i)$ with respect to a target reference point, where $x_i$ and $y_i$ denote the range and cross-range offsets, respectively. We assume that sufficient isolation and self-interference cancellation are available between the transmit and receive chains. Hence, the residual self-interference from the transmit array to the sensing receiver is neglected.

Let $\mathbf{a}_{\mathrm{t}}(\phi_i,\theta_i)\in\mathbb{C}^{M_t\times 1}$ defined in \eqref{trans_steer} and $\mathbf{a}_{\mathrm{r}}(\phi_i,\theta_i)\in\mathbb{C}^{M_r\times 1}$ denote the transmit and receive steering vectors associated with the $i$-th scattering center, respectively. Under the far-field and narrowband-array assumptions, the received sensing echo vector at the sensing receiver is modeled as the coherent superposition of the echoes from all scattering centers
\begin{align}
\mathbf{y}_{\mathrm{s}}(t)
=&
\sum_{i=1}^{Q}
\beta_i
\mathbf{a}_{\mathrm{r}}(\phi_i,\theta_i)
\mathbf{a}_{\mathrm{t}}^{H}(\phi_i,\theta_i)
\mathbf{x}\big(t-\tau_i(t)\big)
e^{-j2\pi f_c\tau_i(t)}\nonumber \\ &
+
\mathbf{z}_{\mathrm{s}}(t),
\label{eq:isar_echo_vector}
\end{align}
where $\mathbf{z}_{\mathrm{s}}(t)\in\mathbb{C}^{M_r\times 1}$ denotes the receiver noise, $f_c$ is the carrier frequency, and $\tau_i(t) =\frac{2R_i(t)}{c}$ is the round-trip delay of the $i$-th scattering center. Here, $c$ is the speed of light and $R_i(t)$ denotes the instantaneous range between the radar and the $i$-th scattering center.

The transmit sensing beamformer $\mathbf{w}$ is fixed across all sensing subcarriers and designed offline using a conventional beampattern-synthesis method to steer the mainlobe toward $(\phi_{\mathrm{t}},\theta_{\mathrm{t}})$ while suppressing sidelobe radiation \cite{VanTrees2002OptimumArray,Zhang2018PatternSynthesis}. The beamformer is normalized such that $\|\mathbf{w}\|_2^2=1$. The transmit and receive steering vectors are also unit norm, i.e., $\|\mathbf{a}_{\mathrm{t}}(\phi,\theta)\|_2^2=1$ and $\|\mathbf{a}_{\mathrm{r}}(\phi,\theta)\|_2^2=1$. The sensing receiver employs a conventional matched beamformer steered toward the target reference direction. Thus, the scalar combined sensing echo is given by
\begin{equation}
\widetilde{y}_{\mathrm{s}}(t)
=
\mathbf{a}_{\mathrm{r}}^{H}(\phi_{\mathrm{t}},\theta_{\mathrm{t}})\mathbf{y}_{\mathrm{s}}(t).
\label{eq:rx_combined_echo}
\end{equation}
Substituting \eqref{eq:isar_echo_vector} into \eqref{eq:rx_combined_echo} yields
\begin{equation}
\widetilde{y}_{\mathrm{s}}(t)
=
\sum_{i=1}^{Q}
\bar{\beta}_i
\mathbf{a}_{\mathrm{t}}^{H}(\phi_i,\theta_i)
\mathbf{x}\big(t-\tau_i(t)\big)
e^{-j2\pi f_c\tau_i(t)}
+
\widetilde{z}_{\mathrm{s}}(t),
\label{eq:isar_echo_time}
\end{equation}
where $\bar{\beta}_i \triangleq \beta_i \mathbf{a}_{\mathrm{r}}^{H}(\phi_{\mathrm{t}},\theta_{\mathrm{t}}) \mathbf{a}_{\mathrm{r}}(\phi_i,\theta_i)$ is the effective reflectivity of the $i$-th scattering center after receive beamforming, and $\widetilde{z}_{\mathrm{s}}(t)=\mathbf{a}_{\mathrm{r}}^{H}(\phi_{\mathrm{t}},\theta_{\mathrm{t}})\mathbf{z}_{\mathrm{s}}(t)$ is the combined receiver noise. If $\mathbf{z}_{\mathrm{s}}(t)\sim \mathcal{CN}(\mathbf{0},\sigma_{\mathrm{rad}}^2\mathbf{I}_{M_r})$, then $\widetilde{z}_{\mathrm{s}}(t)\sim \mathcal{CN}(0,\sigma_{\mathrm{rad}}^2)$. Substituting the sensing component of \eqref{eq:x_ln} into \eqref{eq:isar_echo_time}, while collecting the communication-induced term as residual distortion, the receive-combined sensing echo can be expressed as
\begin{align}
\widetilde{y}_{\mathrm{s}}(t)
&=
\frac{1}{\sqrt{N_c}}
\sum_{i=1}^{Q}
\sum_{\ell=0}^{N_s-1}
\sum_{n=0}^{N_c-1}
\beta_i^{\ell,n}
e^{j2\pi n\Delta f (t-\tau_i(t)-\ell T_t-T_g)}
\nonumber\\
&\quad \times
u\!\left(\frac{t-\tau_i(t)-\ell T_t}{T_t}\right)
e^{-j2\pi f_c\tau_i(t)}
+
\widetilde{d}_{\mathrm{c}}(t)
+
\widetilde{z}_{\mathrm{s}}(t),
\label{eq:isar_echo_time2}
\end{align}
where $\beta_i^{\ell,n} \triangleq \bar{\beta}_i \sqrt{q_n\rho_n^{\mathrm{sen}}}\, s_{\ell,n}^{\mathrm{sen}} \mathbf{a}_{\mathrm{t}}^{H}(\phi_i,\theta_i)\mathbf{w}$ is the equivalent sensing coefficient, which includes the effective scatterer reflectivity after receive beamforming, sensing symbol, sensing power, sensing-subcarrier selection, and transmit beamforming gain. The term $\widetilde{d}_{\mathrm{c}}(t)$ denotes the residual communication-induced distortion after receive combining and sensing extraction. With a highly directive sensing beam, sufficiently low sidelobe level, target range gating, and the MASR constraint formulated in Section~\ref{subsec:sensing_metric}, this term can be significantly suppressed and is neglected in the baseline ISAR imaging model.

The target motion is decomposed into translational and rotational components. Let $R_t(t)$ denote the range of the target reference point due to translational motion and let $\Phi(t)=\phi_0+\Omega t$ denote the target rotational motion, where $\phi_0$ is the initial aspect angle and $\Omega$ is the rotational velocity of the target. Under the far-field ISAR approximation, the range of the $i$-th scattering center is given by \cite{Hashempour2020,Hashempour2018}
\begin{equation}
R_i(t)
\simeq
R_t(t)
+
x_i\cos\Phi(t)
-
y_i\sin\Phi(t).
\label{eq:isar_scatterer_range}
\end{equation}
After translational motion compensation (TMC) \cite{Hashempour2017DSP}, the common range variation $R_t(t)$ is removed. Moreover, for a small rotation angle within the coherent processing interval, the scatterer range at slow-time index $\ell$ can be approximated as \cite{Hashempour2017DSP,Hashempour2020}
\begin{equation}
R_i[\ell]
\simeq
R_0
+
x_i
-
y_i\Omega \ell T_t,
\label{eq:motion_final}
\end{equation}
where $R_0$ denotes the reference range after motion compensation. Hence, the cross-range information is encoded in the slow-time phase variation induced by the target rotation.

To obtain the frequency-domain OFDM sensing model, we assume timing synchronization with respect to the target reference range, and that the residual delay spread of the dominant scattering centers is within the CP duration. Hence, the CP absorbs the target-induced delay spread and inter-symbol interference is avoided. Furthermore, after TMC and migration-through-resolution-cells (MTRC) compensation~\cite{Hashempour2020,Xing2004MTRC}, the time-varying envelope term $u((t-\tau_i(t)-\ell T_t)/T_t)$ can be approximated by a fixed OFDM block window over the $\ell$-th symbol. In addition, the target motion within one OFDM symbol is assumed to be sufficiently small such that $\tau_i(t)\approx\tau_i[\ell]=2R_i[\ell]/c$ during the demodulation interval. Therefore, after CP removal, OFDM demodulation, and sensing-symbol extraction \cite{Dai2026MIMOOFDMISAC}, the frequency-domain sensing observation on OFDM symbol $\ell$ and subcarrier $n$ is obtained as
\begin{equation}
Y_{\ell,n}^{\mathrm{sen}}
\approx
\sum_{i=1}^{Q}
\beta_i^{\ell,n}
e^{-j\frac{4\pi}{c}(f_c+n\Delta f)R_i[\ell]}
+
Z_{\ell,n},
\label{eq:isar_echo_freq}
\end{equation}
where $Z_{\ell,n}\sim\mathcal{CN}(0,\sigma_{\mathrm{rad}}^2)$ is the frequency-domain noise. The exponential term captures the range-dependent phase variation across subcarriers and the slow-time phase history required for ISAR imaging.
For active sensing subcarriers, i.e., $n\in\Omega_{\mathrm{sen}}$, the known sensing symbol and power factor can be removed from \eqref{eq:isar_echo_freq}. The known sensing symbols are assumed to be nonzero on the selected sensing subcarriers. Since $\rho_n^{\mathrm{sen}}=1$ for $n\in\Omega_{\mathrm{sen}}$, we define the normalized sensing measurement as $\widetilde{H}_{\ell,n} = \frac{Y_{\ell,n}^{\mathrm{sen}}}{\sqrt{q_n}s_{\ell,n}^{\mathrm{sen}}}$ for $n\in\Omega_{\mathrm{sen}}$. Using the definition of $\beta_i^{\ell,n}$,
and neglecting residual communication-induced
distortion, we obtain
\begin{equation}
\widetilde{H}_{\ell,n}
=
\sum_{i=1}^{Q}
\widetilde{\beta}_i
e^{-j\frac{4\pi}{c}(f_c+n\Delta f)R_i[\ell]}
+E_{\ell,n},
\quad n\in\Omega_{\mathrm{sen}},
\label{eq:channel_final}
\end{equation}
where $\widetilde{\beta}_i \triangleq
\bar{\beta}_i
\mathbf{a}_{\mathrm{t}}^{H}(\phi_i,\theta_i)\mathbf{w}$
denotes the beamformed reflectivity of the $i$-th
scattering center.
Equivalently, for $n\in\Omega_{\mathrm{sen}}$,
$\widetilde{\beta}_i=
\beta_i^{\ell,n}/(\sqrt{q_n}s_{\ell,n}^{\mathrm{sen}})$.
The normalized receiver noise is
$E_{\ell,n}=Z_{\ell,n}/
(\sqrt{q_n}s_{\ell,n}^{\mathrm{sen}})$,
with variance
$\sigma_{\mathrm{rad}}^2/
(q_n|s_{\ell,n}^{\mathrm{sen}}|^2)$.
In the full-band sensing case, where $\Omega_{\mathrm{sen}}=\{0,\ldots,N_c-1\}$, range compression is achieved by applying an inverse discrete Fourier transform (IDFT) across all subcarriers:
\begin{equation}
s_{\ell,q}
=
\frac{1}{N_c}
\sum_{n=0}^{N_c-1}
\widetilde{H}_{\ell,n}
e^{j2\pi nq/N_c},
\label{eq:range_compression_final}
\end{equation}
where $q=0,\ldots,N_c-1$ indexes the range bins. The resulting matrix $\mathbf{S}=[s_{\ell,q}]\in\mathbb{C}^{N_s\times N_c}$ represents the range-compressed ISAR data, where rows correspond to slow-time samples and columns correspond to range bins. The ISAR image is then obtained by applying a DFT along the slow-time dimension
\begin{equation}
I(p,q)
=
\sum_{\ell=0}^{N_s-1}
s_{\ell,q}
e^{-j2\pi p\ell/N_s},
\quad
p=0,\ldots,N_s-1 .
\label{eq:isar_image_final}
\end{equation}
Here, $p$ denotes the Doppler, or cross-range, index. This is the conventional two-dimensional fast Fourier transform (2D-FFT) approach for ISAR imaging with complete data \cite{Hashempour2017DSP,Hashempour2018,Hashempour2020}.

In the sparse sensing case, only the measurements $\{\widetilde{H}_{\ell,n}:n\in\Omega_{\mathrm{sen}}\}$ are available. Therefore, the IDFT in \eqref{eq:range_compression_final} cannot be directly applied to the complete frequency-domain data. Instead, the missing frequency-domain samples, or equivalently the sparse ISAR image, are recovered using the compressed-sensing reconstruction method described in the next subsection.
\vspace{-2mm}

\subsection{Sparse Aperture and Sparse Subcarrier Sensing in ISAC Systems}

In practical TDD-based ISAC systems, not all OFDM symbols are available for sensing. Let $T_c$ denote the coherence time, measured in seconds. Since each OFDM symbol has duration $T_t$, the number of OFDM symbols contained in one coherence interval is
\begin{equation}
\tau_c
=
\left\lfloor
\frac{T_c}{T_t}
\right\rfloor .
\end{equation}
Following the TDD protocol, $\tau_p$ OFDM symbols are reserved for uplink pilot training, while the remaining $\tau_c-\tau_p$ OFDM symbols are used for downlink data transmission and sensing as shown in Fig.~\ref{fig:tdd_sparse_aperture}. Therefore, only a subset of the nominal slow-time samples is available for ISAR processing. The fraction of each coherence interval available for downlink communication and sensing is
\begin{equation}\label{eq:dl_fraction}
\zeta_{\mathrm{dl}}
=
1-\frac{\tau_p}{\tau_c}.
\end{equation}
Consider $N_{\mathrm{blk}}$ consecutive coherence blocks. The nominal slow-time aperture then contains
\begin{equation}
N_s
=
N_{\mathrm{blk}}\tau_c
\end{equation}
OFDM-symbol positions. For the $b$-th coherence block, $b=0,\ldots,N_{\mathrm{blk}}-1$, define the block index set as $\mathcal{T}_b = \{b\tau_c,\ldots,(b+1)\tau_c-1\}$. Let $\mathcal{P}_b\subset \mathcal{T}_b$ denote the set of uplink pilot-symbol indices, with $|\mathcal{P}_b|=\tau_p$. The available downlink sensing-symbol indices in block $b$ are therefore given by
\begin{equation}
\mathcal{S}_b
=
\mathcal{T}_b\setminus \mathcal{P}_b,
\qquad
|\mathcal{S}_b|=\tau_c-\tau_p .
\end{equation}
The overall set of available slow-time samples is
\begin{equation}
\mathcal{S}
=
\bigcup_{b=0}^{N_{\mathrm{blk}}-1}
\mathcal{S}_b,
\end{equation}
and the number of available slow-time samples is
\begin{equation}
\widetilde{N}_s
=
|\mathcal{S}|
=
N_{\mathrm{blk}}(\tau_c-\tau_p)
<
N_s .
\end{equation}
The slow-time index $\ell$ in \eqref{eq:motion_final}
denotes the absolute OFDM-symbol position.
\begin{figure*}[!t]
\centering
\begin{tikzpicture}[
    font=\footnotesize,
    ul/.style={draw, fill=orange!35, minimum height=0.7cm, text width=1.15cm, align=center, inner sep=0pt},
    dl/.style={draw, fill=blue!30, minimum height=0.7cm, text width=2.25cm, align=center, inner sep=0pt},
    miss/.style={draw, fill=gray!35, minimum height=0.7cm, text width=1.15cm, align=center, inner sep=0pt},
    valid/.style={draw, fill=white, minimum height=0.7cm, text width=2.25cm, align=center, inner sep=0pt},
    >=latex
]

\node[ul] (ul1) at (0,0) {UL pilots};
\node[dl, right=0cm of ul1] (dl1) {DL data+sensing};
\node[ul, right=0cm of dl1] (ul2) {UL pilots};
\node[dl, right=0cm of ul2] (dl2) {DL data+sensing};
\node[ul, right=0cm of dl2] (ul3) {UL pilots};

\node[right=0.55cm of ul3] (dots1) {$\cdots$};

\node[dl, right=0.55cm of dots1] (dlB1) {DL data+sensing};
\node[ul, right=0cm of dlB1] (ulB) {UL pilots};
\node[dl, right=0cm of ulB] (dlB2) {DL data+sensing};

\node[miss, below=0.85cm of ul1] (m1) {Missing};
\node[valid, below=0.85cm of dl1] (v1) {Valid};
\node[miss, below=0.85cm of ul2] (m2) {Missing};
\node[valid, below=0.85cm of dl2] (v2) {Valid};
\node[miss, below=0.85cm of ul3] (m3) {Missing};

\node[below=1.15cm of dots1] (dots2) {$\cdots$};

\node[valid, below=0.85cm of dlB1] (vB1) {Valid};
\node[miss, below=0.85cm of ulB] (mB) {Missing};
\node[valid, below=0.85cm of dlB2] (vB2) {Valid};

\draw[<->, thick]
([yshift=0.58cm]ul1.north west)
--
([yshift=0.58cm]dlB2.north east)
node[midway, above=0.05cm] {TDD communication protocol};

\draw[<->, thick]
([xshift=-0.0cm,yshift=0.58cm]ul1.north west)
--
([xshift=0.0cm,yshift=0.58cm]dlB2.north east);

\draw[<->, thick]
([xshift=-0.00cm,yshift=-0.30cm]m1.south west)
--
([xshift=0.00cm,yshift=-0.30cm]vB2.south east)
node[midway, below=0.15cm] {Slow time / ISAR aperture};

\draw[<->, dashed]
([yshift=0.20cm]ul1.north west)
--
([yshift=0.20cm]dl1.north east)
node[midway, above=0.02cm] {$\tau_c$};

\draw[<->]
([yshift=-0.25cm]ul1.south west)
--
([yshift=-0.25cm]ul1.south east)
node[midway, below=0.05cm] {$\tau_p$};

\draw[<->]
([yshift=-0.25cm]dl1.south west)
--
([yshift=-0.25cm]dl1.south east)
node[midway, below=0.05cm] {$\tau_c-\tau_p$};

\foreach \x in {ul1,dl1,ul2,dl2,ul3,dlB1,ulB,dlB2}{
    \draw[dashed, gray!70]
    (\x.south west) -- ++(0,-1.05cm);
    \draw[dashed, gray!70]
    (\x.south east) -- ++(0,-1.05cm);
}

\end{tikzpicture}
\caption{Sparse slow-time aperture in TDD mMIMO ISAC. UL pilot symbols occupy $\tau_p$ out of $\tau_c$ OFDM symbols per coherence interval and create missing sensing samples, while the remaining DL symbols provide valid ISAR observations.}
\label{fig:tdd_sparse_aperture}
\end{figure*}
From an ISAR perspective, the missing uplink-pilot intervals create a sparse slow-time aperture, resulting in incomplete Doppler-domain information. Moreover, in the considered ISAC waveform, sensing may be performed only over a subset of subcarriers, i.e., $n\in\Omega_{\mathrm{sen}}$, to enable sparse waveform-domain sensing. Thus, the ISAR observation may be incomplete in both the slow-time and frequency dimensions.
Let $\Omega_{\mathrm{sen}} = \{n_1,\ldots,n_{\widetilde{N}_c}\} \subseteq \{0,\ldots,N_c-1\}$ denote the active sensing-subcarrier set, where $\widetilde{N}_c=|\Omega_{\mathrm{sen}}|$. To model the sparse slow-time aperture, we define the aperture sampling matrix $\boldsymbol{\Psi}_a\in\{0,1\}^{\widetilde{N}_s\times N_s}$ as
\begin{equation}
\boldsymbol{\Psi}_a
=
\begin{bmatrix}
\left(\mathbf{i}^{(N_s)}_{\ell_1+1}\right)^T\\
\left(\mathbf{i}^{(N_s)}_{\ell_2+1}\right)^T\\
\vdots\\
\left(\mathbf{i}^{(N_s)}_{\ell_{\widetilde{N}_s}+1}\right)^T
\end{bmatrix},
\qquad
\{\ell_1,\ldots,\ell_{\widetilde{N}_s}\}=\mathcal{S},
\end{equation}
where $\mathbf{i}^{(N_s)}_{r}$ denotes the $r$-th column of the $N_s\times N_s$ identity matrix. Similarly, to model sparse sensing subcarriers, we define the subcarrier sampling matrix $\boldsymbol{\Psi}_r\in\{0,1\}^{\widetilde{N}_c\times N_c}$ as
\begin{equation}
\boldsymbol{\Psi}_r
=
\begin{bmatrix}
\left(\mathbf{i}^{(N_c)}_{n_1+1}\right)^T\\
\left(\mathbf{i}^{(N_c)}_{n_2+1}\right)^T\\
\vdots\\
\left(\mathbf{i}^{(N_c)}_{n_{\widetilde{N}_c}+1}\right)^T
\end{bmatrix},
\qquad
\{n_1,\ldots,n_{\widetilde{N}_c}\}=\Omega_{\mathrm{sen}},
\end{equation}
where $\mathbf{i}^{(N_c)}_{r}$ denotes the $r$-th column of the $N_c\times N_c$ identity matrix. In the proposed joint resource-allocation framework, the active sensing-subcarrier set $\Omega_{\mathrm{sen}}$ is determined by the SAC policy described in Section~\ref{sec:DRL}. Let $\widetilde{\mathbf{H}}\in\mathbb{C}^{N_s\times N_c}$ denote the hypothetical complete normalized frequency-domain ISAR measurement matrix whose $(\ell,n)$-th entry follows the model in \eqref{eq:channel_final}. The complete frequency-domain ISAR data can be modeled as
\begin{equation}
\widetilde{\mathbf{H}}
=
\mathbf{F}_a
\mathbf{G}
\mathbf{F}_r^{T}
+
\mathbf{Z},
\label{eq:full_2d_isar_model}
\end{equation}
where $\mathbf{G}\in\mathbb{C}^{N_a\times N_r}$
denotes the discretized two-dimensional ISAR reflectivity
image on the cross-range--range grid.
Matrices $\mathbf{F}_a\in\mathbb{C}^{N_s\times N_a}$
and $\mathbf{F}_r\in\mathbb{C}^{N_c\times N_r}$
are the slow-time and frequency-domain Fourier
dictionaries, respectively,
and $\mathbf{Z}$ collects the normalized
noise terms $E_{\ell,n}$ defined in
\eqref{eq:channel_final}. The Fourier dictionaries are defined as
\begin{subequations}
    \allowdisplaybreaks
    \begin{align}
[\mathbf{F}_a]_{\ell,m}
&=
\frac{1}{\sqrt{N_a}}e^{-j2\pi \ell m/N_a},
\nonumber \\&
\ell=0,\ldots,N_s-1,\quad m=0,\ldots,N_a-1,
\\
[\mathbf{F}_r]_{n,q}
&=
\frac{1}{\sqrt{N_r}}e^{-j2\pi n q/N_r},
\nonumber \\ &
n=0,\ldots,N_c-1,\quad q=0,\ldots,N_r-1.
\end{align}
\end{subequations}
Due to the missing slow-time samples and the sparse sensing subcarriers, the available compressed ISAR observation is
\begin{equation}
\mathbf{S}_{\mathrm{CS}}
=
\boldsymbol{\Psi}_a
\widetilde{\mathbf{H}}
\boldsymbol{\Psi}_r^{T}
=
\boldsymbol{\Psi}_a
\mathbf{F}_a
\mathbf{G}
\mathbf{F}_r^{T}
\boldsymbol{\Psi}_r^{T}
+
\mathbf{Z}_{\mathrm{CS}},
\label{eq:2d_sparse_isar_model}
\end{equation}
where $\mathbf{S}_{\mathrm{CS}}\in
\mathbb{C}^{\widetilde{N}_s\times\widetilde{N}_c}$
is the observed sparse ISAR data matrix and
$\mathbf{Z}_{\mathrm{CS}}\in
\mathbb{C}^{\widetilde{N}_s\times\widetilde{N}_c}$
is the corresponding normalized noise matrix,
$\mathbf{Z}_{\mathrm{CS}}
=\mathbf{\Psi}_a\mathbf{Z}\mathbf{\Psi}_r^T$,
whose entries have variance
$\sigma_{\mathrm{rad}}^2/
(q_n|s_{\ell,n}^{\mathrm{sen}}|^2)$
at the selected sensing indices.
The dimensions $N_a$ and $N_r$ denote the numbers of discretized cross-range and range grid points, respectively. For super-resolution ISAR imaging, the image grid can be finer than the number of available measurements, i.e., $N_a>N_s$ and/or $N_r>N_c$. Therefore, recovering $\mathbf{G}$ from \eqref{eq:2d_sparse_isar_model} is an underdetermined inverse problem, which motivates the sparse reconstruction method developed in the next subsection.
\vspace{-2mm}

\subsection{Sparse ISAR Image Reconstruction From Incomplete ISAC Measurements}
\label{subsec:2d_sparse_reconstruction}

Based on the sparse aperture and sparse subcarrier sensing model developed in the previous subsection, the ISAR image is reconstructed by estimating the reflectivity matrix $\mathbf{G}$ from \eqref{eq:2d_sparse_isar_model}. Since an ISAR image is typically dominated by a limited number of strong scattering centers, $\mathbf{G}$ can be regarded as sparse in the discretized range--Doppler domain. Accordingly, we formulate the sparse reconstruction problem as
\begin{equation}
\widehat{\mathbf{G}}
=
\arg\min_{\mathbf{G}}
\frac{1}{2}
\left\|
\mathbf{A}\mathbf{G}\mathbf{C}
-
\mathbf{S}_{\mathrm{CS}}
\right\|_F^2
+
\lambda \|\mathbf{G}\|_1 ,
\label{eq:compact_l1_problem}
\end{equation}
where $\lambda>0$ is the sparsity-promoting regularization parameter, $\|\mathbf{G}\|_1$ denotes the element-wise $\ell_1$ norm, $\mathbf{A} \triangleq \boldsymbol{\Psi}_a\mathbf{F}_a$, and $\mathbf{C} \triangleq \mathbf{F}_r^{T}\boldsymbol{\Psi}_r^{T}$. The unweighted formulation is adopted for computational efficiency and is well suited to uniform-power sensing. Under nonuniform power allocation, it remains applicable, although noise-weighted reconstruction may improve accuracy. A conventional 2D-ADMM solution to \eqref{eq:compact_l1_problem} introduces an auxiliary variable $\mathbf{B}$ and imposes $\mathbf{G}=\mathbf{B}$, leading to
\begin{equation}
\min_{\mathbf{G},\mathbf{B}}
\frac{1}{2}
\left\|
\mathbf{A}\mathbf{G}\mathbf{C}
-
\mathbf{S}_{\mathrm{CS}}
\right\|_F^2
+
\lambda \|\mathbf{B}\|_1,
\quad
\mathrm{s.t.}\quad
\mathbf{G}=\mathbf{B}.
\label{eq:admm_constrained_problem}
\end{equation}
Using the scaled dual variable $\mathbf{V}$ and ADMM penalty parameter $\delta>0$, the data-fitting residual at iteration $j$ is
\begin{equation}
\mathbf{R}_{s}^{j}
=
\boldsymbol{\Psi}_a\mathbf{F}_a
(\mathbf{B}^{j}-\mathbf{V}^{j})
\mathbf{F}_r^{T}\boldsymbol{\Psi}_r^{T}
-
\mathbf{S}_{\mathrm{CS}} .
\label{eq:admm_residual}
\end{equation}
The reflectivity matrix is then updated as \cite{Hashempour2020}
\begin{equation}
\mathbf{G}^{j+1}
=
(\mathbf{B}^{j}-\mathbf{V}^{j})
-
\frac{1}{1+\delta}
\mathbf{F}_a^{H}\boldsymbol{\Psi}_a^{T}
\mathbf{R}_{s}^{j}
\boldsymbol{\Psi}_r
\mathbf{F}_r^{*}.
\label{eq:g_admm_update_simplified}
\end{equation}
The auxiliary variable and scaled dual variable are updated as
\begin{align}
\mathbf{B}^{j+1}
&=
\mathcal{S}_{\lambda/\delta}
\left(
\mathbf{G}^{j+1}
+
\mathbf{V}^{j}
\right),
\label{eq:b_admm_update}
\\
\mathbf{V}^{j+1}
&=
\mathbf{V}^{j}
+
\mathbf{G}^{j+1}
-
\mathbf{B}^{j+1},
\label{eq:v_admm_update}
\end{align}
where the element-wise complex soft-thresholding
operator is defined as
\begin{equation}
\mathcal{S}_{\vartheta}(x)
=
\max\left(1-\frac{\vartheta}{|x|},0\right)x,
\end{equation}
with $\mathcal{S}_{\vartheta}(0)=0$
and $\vartheta\geq0$.
Although the above conventional 2D-ADMM provides an efficient reconstruction baseline, its image quality depends strongly on the manually selected regularization parameter $\lambda$. A large $\lambda$ may suppress weak scatterers, whereas a small $\lambda$ may retain sidelobe-like artifacts and background noise. Moreover, a fixed ADMM penalty parameter $\delta$ may lead to slow or unbalanced convergence. To address these limitations, we propose an adaptive lobe-suppressing reweighted 2D-ADMM method. The key idea is to estimate the sparsity threshold from the preliminary ISAR image itself and to update the ADMM penalty parameter according to the primal and dual residuals.
First, a preliminary matched-filter image is formed as
\begin{equation}
\widetilde{\mathbf{G}}
=
\mathbf{A}^{H}
\mathbf{S}_{\mathrm{CS}}
\mathbf{C}^{H}
=
\mathbf{F}_a^{H}\boldsymbol{\Psi}_a^{T}
\mathbf{S}_{\mathrm{CS}}
\boldsymbol{\Psi}_r
\mathbf{F}_r^{*}.
\label{eq:dirty_isar_image}
\end{equation}
Let $a_1 \geq a_2 \geq \cdots \geq a_{N_G}$ denote the sorted magnitudes of the vectorized coefficients of $\widetilde{\mathbf{G}}$, where $N_G$ is the number of image pixels. Since dominant scattering centers are usually separated from the sidelobe/noise background in magnitude, we estimate the dominant support size as
\begin{equation}
\widehat{Q}
=
\arg\max_{q\in\mathcal{Q}}
\log
\left(
\frac{a_q+\epsilon_g}{a_{q+1}+\epsilon_g}
\right),
\label{eq:gap_support_estimation}
\end{equation}
where $\mathcal{Q}=\{1,\ldots,Q_{\max}\}$ is the admissible support-search set and $\epsilon_g>0$ avoids numerical instability. The corresponding gap-based threshold is
\begin{equation}
\lambda_{\mathrm{gap}}
=
\frac{1}{2}
\left(
a_{\widehat{Q}}+a_{\widehat{Q}+1}
\right).
\label{eq:lambda_gap}
\end{equation}
To avoid an overly small threshold in noisy cases, we also compute
\begin{equation}
\lambda_{\mathrm{noise}}
=
c_{\mathrm{noise}}\widehat{\sigma}
\sqrt{2\log N_G},
\label{eq:lambda_noise}
\end{equation}
where $c_{\mathrm{noise}}>0$ and $\widehat{\sigma}$ is a robust estimate of the background level obtained from the small-magnitude coefficients $\mathcal{M}_{\mathrm{bg}}$ of $\widetilde{\mathbf{G}}$ via $\widehat{\sigma}=\operatorname{median}_{m\in\mathcal{M}_{\mathrm{bg}}}(a_m)/\sqrt{2\log 2}$. The target regularization parameter is then selected as
\begin{equation}
\lambda_{\mathrm{tar}}
=
\max
\left\{
\lambda_{\mathrm{noise}},
c_{\mathrm{gap}}\lambda_{\mathrm{gap}}
\right\},
\label{eq:lambda_target}
\end{equation}
where $0<c_{\mathrm{gap}}\leq 1$. For stability, a continuation strategy is used:
\begin{equation}
\lambda_j
=
\max
\left\{
\lambda_{\mathrm{tar}},
\lambda_0\eta^j
\right\},
\qquad 0<\eta<1.
\label{eq:lambda_continuation}
\end{equation}
To further enhance sparsity, the fixed $\ell_1$ penalty in \eqref{eq:compact_l1_problem} is replaced by a reweighted $\ell_1$ term:
\begin{equation}
\widehat{\mathbf{G}}
=
\arg\min_{\mathbf{G}}
\frac{1}{2}
\left\|
\mathbf{A}\mathbf{G}\mathbf{C}
-
\mathbf{S}_{\mathrm{CS}}
\right\|_F^2
+
\lambda_j
\left\|
\mathbf{W}^{j}\odot\mathbf{G}
\right\|_1 ,
\label{eq:weighted_l1_problem}
\end{equation}
where $\mathbf{W}^{j}$ is the reweighting matrix and $\odot$ denotes the Hadamard product. The weights are updated as
\begin{equation}
\left[\mathbf{W}^{j+1}\right]_{u,v}
=
\frac{1}{\left|[\mathbf{B}^{j+1}]_{u,v}\right|+\epsilon_w},
\label{eq:weight_update}
\end{equation}
followed by normalization with respect to their median value. This reweighting penalizes weak coefficients more strongly than dominant scatterers and therefore improves lobe and background suppression.
The proposed method preserves the computational structure of conventional 2D-ADMM. Specifically, the $\mathbf{G}$- and $\mathbf{V}$-updates follow \eqref{eq:g_admm_update_simplified} and \eqref{eq:v_admm_update}, while the auxiliary-variable update becomes
\begin{equation}
\mathbf{B}^{j+1}
=
\mathcal{S}_{(\lambda_j/\delta_j)\mathbf{W}^{j}}
\left(
\mathbf{G}^{j+1}
+
\mathbf{V}^{j}
\right).
\label{eq:weighted_b_admm_update}
\end{equation}
The ADMM penalty parameter is updated using the primal and dual residuals
\begin{equation}
\begin{aligned}
\mathbf{R}_{\mathrm{pri}}^{j+1}
&=
\mathbf{G}^{j+1}-\mathbf{B}^{j+1},\\
\mathbf{R}_{\mathrm{dual}}^{j+1}
&=
\delta_j
\left(
\mathbf{B}^{j+1}-\mathbf{B}^{j}
\right),
\end{aligned}
\label{eq:primal_dual_residuals}
\end{equation}
according to
\begin{equation}
\delta_{j+1}
=
\begin{cases}
\tau_{\delta}\delta_j, &
\|\mathbf{R}_{\mathrm{pri}}^{j+1}\|_F
>
\mu
\|\mathbf{R}_{\mathrm{dual}}^{j+1}\|_F,\\
\delta_j/\tau_{\delta}, &
\|\mathbf{R}_{\mathrm{dual}}^{j+1}\|_F
>
\mu
\|\mathbf{R}_{\mathrm{pri}}^{j+1}\|_F,\\
\delta_j, & \text{otherwise},
\end{cases}
\label{eq:delta_update}
\end{equation}
where $\mu>1$ and $\tau_{\delta}>1$. When $\delta_j$ is changed, the scaled dual variable is rescaled as $\mathbf{V}^{j+1}\leftarrow (\delta_j/\delta_{j+1})\mathbf{V}^{j+1}$ to preserve the unscaled dual variable. The conventional 2D-ADMM method \cite{Hashempour2020} is recovered as a special case by setting $\mathbf{W}^{j}=\mathbf{1}$, $\lambda_j=\lambda$, and $\delta_j=\delta$ for all iterations. The loop terminates when $\|\mathbf{R}_{\mathrm{pri}}^{j+1}\|_F/\max(\|\mathbf{G}^{j+1}\|_F,1)\leq \epsilon_{\mathrm{tol}}$ or $j=J_{\max}$, with $\mathcal{O}(N_a N_r \log(N_a N_r))$ per-iteration 2D-FFT complexity. The proposed reconstruction procedure is summarized in Algorithm~\ref{alg:adaptive-admm-clean}.

\begin{algorithm}[t]
\caption{Sparse ISAR Image Reconstruction via Adaptive Reweighted 2D-ADMM}
\label{alg:adaptive-admm-clean}

\KwIn{$\mathbf{S}_{\mathrm{CS}}$, $\boldsymbol{\Psi}_a$, $\mathbf{F}_a$, $\mathbf{F}_r$, $\boldsymbol{\Psi}_r$, $\delta_0$, $\lambda_0$, $c_{\mathrm{noise}}$, $c_{\mathrm{gap}}$, $\eta$, $\mu$, and $\tau_{\delta}$}

\KwOut{$\widehat{\mathbf{G}}$}

Form $\widetilde{\mathbf{G}}$ using \eqref{eq:dirty_isar_image}\;

Estimate $\widehat{Q}$ using \eqref{eq:gap_support_estimation}\;

Compute $\lambda_{\mathrm{gap}}$, $\lambda_{\mathrm{noise}}$, and $\lambda_{\mathrm{tar}}$ using \eqref{eq:lambda_gap}--\eqref{eq:lambda_target}\;

Initialize $j=0$, $\mathbf{G}^{0}=\mathbf{0}$, $\mathbf{B}^{0}=\mathbf{0}$, $\mathbf{V}^{0}=\mathbf{0}$, $\mathbf{W}^{0}=\mathbf{1}$, and $\delta_j=\delta_0$\;

\While{stopping criterion is not satisfied}{
    Update $\lambda_j$ using \eqref{eq:lambda_continuation}\;

    Compute $\mathbf{R}_s^j$ using \eqref{eq:admm_residual}\;

    Update $\mathbf{G}^{j+1}$ using \eqref{eq:g_admm_update_simplified} with $\delta=\delta_j$\;

    Update $\mathbf{B}^{j+1}$ using \eqref{eq:weighted_b_admm_update}\;

    Update $\mathbf{V}^{j+1}$ using \eqref{eq:v_admm_update}\;

    Update $\delta_{j+1}$ using \eqref{eq:primal_dual_residuals} and \eqref{eq:delta_update}, and rescale $\mathbf{V}^{j+1}$\;

    Update and normalize $\mathbf{W}^{j+1}$ using \eqref{eq:weight_update}\;

    $j\leftarrow j+1$\;
}

$\widehat{\mathbf{G}}=\mathbf{B}^{j}$\;

\end{algorithm}

\section{Performance Metrics and Problem Formulation}\label{sec:metrics}

\subsection{Communication Performance Metrics}\label{subsec:sinr}

We adopt the use-and-then-forget (UatF) bound to characterize the achievable downlink rates under imperfect CSI \cite{Marzetta2016}. For CU $k$ on OFDM symbol $\ell$ and subcarrier $n$, the received signal is
\begin{equation}
y_{k,\ell,n}
=
\mathbf{g}_k^{H}[n]\mathbf{x}_{\ell,n}
+
v_{k,\ell,n},
\label{eq:cu_rx_signal}
\end{equation}
where $v_{k,\ell,n}\sim\mathcal{CN}(0,\sigma_{k,n}^{2})$ is the noise. The per-subcarrier noise variance is modeled as $\sigma_{k,n}^{2} = \frac{\sigma_k^{2}}{N_c}$, where $\sigma_k^2$ denotes the total noise power over the system bandwidth. The transmit vector $\mathbf{x}_{\ell,n}$ is given in \eqref{eq:x_ln}, and the channel decomposition follows \eqref{eq:ch_decomp}.

\begin{Proposition}\label{prop:mr_uatf}
For MR precoding, the UatF effective SINR of CU $k$ on subcarrier $n$ is given by
\begin{equation}
\mathrm{SINR}^{\mathrm{eff,MR}}_{k,n}
=
\frac{
M_t\eta_k\,\rho_{k,n}\alpha_{k,n}
}{
\xi_k\sum\limits_{i=1}^{K}\rho_{i,n}\alpha_{i,n}
+
\xi_k\rho_n^{\mathrm{sen}}q_n
+
\sigma_{k,n}^2
}.
\label{eq:sinr_mr_ofdma}
\end{equation}
\end{Proposition}

\begin{proof}
Please refer to Appendix~\ref{app:proof_mr_uatf}.
\end{proof}

\begin{Proposition}\label{prop:zf_uatf}
For ZF precoding, the UatF effective SINR of CU $k$ on subcarrier $n$ is given by
\begin{equation}
\mathrm{SINR}^{\mathrm{eff,ZF}}_{k,n}
=
\frac{
(M_t-L_n)\eta_k\,\rho_{k,n}\alpha_{k,n}
}{
\varepsilon_k\sum\limits_{i=1}^{K}\rho_{i,n}\alpha_{i,n}
+
\xi_k\rho_n^{\mathrm{sen}}q_n
+
\sigma_{k,n}^2
}.
\label{eq:sinr_zf_ofdma}
\end{equation}
\end{Proposition}

\begin{proof}
Please refer to Appendix~\ref{app:proof_zf_uatf}.
\end{proof}
Since a CU may be assigned multiple subcarriers, the achievable downlink rate of CU $k$, accounting for the uplink channel-estimation overhead, is given by
\begin{equation}
R_k
=
\zeta_{\mathrm{dl}}
\Delta f
\sum_{n=0}^{N_c-1}
\rho_{k,n}
\log_2\!\left(
1+\mathrm{SINR}^{\mathrm{eff}}_{k,n}
\right),
\label{eq:CU_rate}
\end{equation}
where $\zeta_{\mathrm{dl}}$ is defined in
\eqref{eq:dl_fraction}, and
$\mathrm{SINR}^{\mathrm{eff}}_{k,n}$ is given by
\eqref{eq:sinr_mr_ofdma} for MR precoding or by
\eqref{eq:sinr_zf_ofdma} for ZF precoding.
Owing to favorable propagation in massive MIMO,
linear precoding can achieve near-optimal performance.
Among linear schemes, MR and ZF represent the
limiting cases of regularized ZF (RZF), approaching its performance
at low and moderate-to-high SNRs, respectively
\cite{marzetta2016fundamentals}.

\subsection{Sensing Performance Metrics}
\label{subsec:sensing_metric}

Let $(\phi_{\mathrm{t}},\theta_{\mathrm{t}})$ denote the estimated reference direction of the extended target. A coarse angular estimate is assumed to be available from an initial radar search or tracking stage. Since the target direction changes negligibly over the coherent processing interval, it is treated as known and fixed during ISAR imaging. The normalized sensing beampattern gain toward the target is defined as
\begin{equation}
p_0
\triangleq
\left|
\mathbf{a}_{\mathrm{t}}^{H}(\phi_{\mathrm{t}},\theta_{\mathrm{t}})
\mathbf{w}
\right|^2,
\qquad
0\leq p_0\leq1.
\label{eq:p0_def}
\end{equation}
We assume $p_0>0$, corresponding to nonzero illumination in the target direction. To suppress undesired illumination, a sidelobe-level constraint can be imposed as
\begin{equation}
\left|
\mathbf{a}_{\mathrm{t}}^{H}(\phi,\theta)\mathbf{w}
\right|^2
\leq
\gamma_{\mathrm{SL}}p_0,
\qquad
(\phi,\theta)\in\mathcal{G}_{\mathrm{SL}},
\label{eq:sll_constraint}
\end{equation}
where $\mathcal{G}_{\mathrm{SL}}$ denotes the sidelobe angular region and $\gamma_{\mathrm{SL}}$ is the maximum allowable sidelobe-to-mainlobe ratio.

Since the sensing waveform is transmitted only over the active sensing subcarriers, the per-subcarrier sensing SNR in the target direction is
\begin{equation}
\mathrm{SNR}^{\mathrm{sen}}_{n}
=
\frac{
\rho_n^{\mathrm{sen}}q_n
\left|
\mathbf{a}_t^{H}(\phi_{\mathrm{t}},\theta_{\mathrm{t}})
\mathbf{w}
\right|^2
}{
\sigma_{\mathrm{rad}}^2
}
=
\frac{
\rho_n^{\mathrm{sen}}q_n p_0
}{
\sigma_{\mathrm{rad}}^2
},
\quad \forall n\in\mathcal{N}_c,
\label{eq:snr_sen_rb}
\end{equation}
where $\sigma_{\mathrm{rad}}^2$ denotes the sensing receiver noise variance on each subcarrier and $\mathcal{N}_c \triangleq \{0,\ldots,N_c-1\}$. Hence, $\mathrm{SNR}^{\mathrm{sen}}_{n}=0$ for inactive sensing subcarriers with $\rho_n^{\mathrm{sen}}=0$. This is a normalized illumination-to-noise metric; the actual received echo SNR additionally depends on the target reflectivity, propagation loss, and receive beamforming gain, incorporated in $\beta_i$ and $\bar{\beta}_i$. For a given angular direction $(\phi,\theta)$, the sensing-only spatial power over the active sensing subcarriers is
\begin{align}
P_{\mathrm{sen}}(\phi,\theta)
&\triangleq
\mathbb{E}
\left[
\sum_{n=0}^{N_c-1}
\left|
\mathbf{a}_{\mathrm{t}}^H(\phi,\theta)
\sqrt{\rho_n^{\mathrm{sen}}q_n}
\mathbf{w}
s_{\ell,n}^{\mathrm{sen}}
\right|^2
\right]
\nonumber\\
&=
\sum_{n=0}^{N_c-1}
\rho_n^{\mathrm{sen}}q_n
\left|
\mathbf{a}_{\mathrm{t}}^H(\phi,\theta)\mathbf{w}
\right|^2.
\label{eq:Psen_ofdm}
\end{align}
In particular, the sensing power toward the target direction is
\begin{equation}
P_{\mathrm{sen}}(\phi_{\mathrm{t}},\theta_{\mathrm{t}})
=
p_0
\sum_{n=0}^{N_c-1}
\rho_n^{\mathrm{sen}}q_n.
\label{eq:Psen_target}
\end{equation}

The simultaneous downlink communication transmission produces additional spatial leakage in the sensing direction. Since only scheduled CUs transmit on subcarrier $n$, the average communication-induced distortion power on subcarrier $n$ is defined as
\begin{equation}
P_{\mathrm{com}}^{\mathrm{ave}}(\phi,\theta;n)
\triangleq
\mathbb{E}
\left[
\left\|
\mathbf{D}_n
\mathbf{T}_{\mathrm{c}}^{H}[n]
\mathbf{a}_{\mathrm{t}}(\phi,\theta)
\right\|_2^2
\right],
\label{eq:Pcom_ave_def}
\end{equation}
where the expectation is taken over the small-scale fading embedded in the channel estimates used to construct $\mathbf{T}_{\mathrm{c}}[n]$.

For MR and ZF precoding under i.i.d. Rayleigh channel estimates and a unit-norm steering vector $\mathbf{a}_{\mathrm{t}}(\phi,\theta)$, the average distortion power satisfies \cite{Liao2024mMIMOISAC}
\begin{equation}
P_{\mathrm{com}}^{\mathrm{ave}}(\phi,\theta;n)
=
\frac{1}{M_t}
\sum_{k=1}^{K}
\rho_{k,n}\alpha_{k,n},
\label{eq:Pcom_M_scaling}
\end{equation}
which is independent of $(\phi,\theta)$ due to isotropy.

Since sensing observations are collected only on the active sensing subcarriers, the total communication-induced distortion over the sensing band is defined as
\begin{align}
P_{\mathrm{com}}^{\mathrm{tot}}
&\triangleq
\sum_{n=0}^{N_c-1}
\rho_n^{\mathrm{sen}}
P_{\mathrm{com}}^{\mathrm{ave}}(\phi,\theta;n)
\nonumber\\
&=
\frac{1}{M_t}
\sum_{n=0}^{N_c-1}
\rho_n^{\mathrm{sen}}
\sum_{k=1}^{K}
\rho_{k,n}\alpha_{k,n}.
\label{eq:Pcom_ave_avg}
\end{align}
For $P_{\mathrm{com}}^{\mathrm{tot}}>0$, using \eqref{eq:Psen_target} and \eqref{eq:Pcom_ave_avg}, the MASR is defined as
\begin{align}
\mathrm{MASR}
&\triangleq
\frac{
P_{\mathrm{sen}}(\phi_{\mathrm{t}},\theta_{\mathrm{t}})
}{
P_{\mathrm{com}}^{\mathrm{tot}}
}
\nonumber\\
&=
\frac{
M_t p_0
\sum_{n=0}^{N_c-1}\rho_n^{\mathrm{sen}}q_n
}{
\sum_{n=0}^{N_c-1}
\rho_n^{\mathrm{sen}}
\sum_{k=1}^{K}
\rho_{k,n}\alpha_{k,n}
}.
\label{eq:masr_M_final}
\end{align}
The MASR compares sensing illumination with average spatial communication leakage, whose received contribution also depends on propagation and receive processing. Imposing $\mathrm{MASR}\geq\kappa$, where $\kappa$ is the required MASR threshold, is equivalent to
\begin{equation}
M_t p_0
\sum_{n=0}^{N_c-1}
\rho_n^{\mathrm{sen}}q_n
\geq
\kappa
\sum_{n=0}^{N_c-1}
\rho_n^{\mathrm{sen}}
\sum_{k=1}^{K}
\rho_{k,n}\alpha_{k,n}.
\label{eq:masr_constr_M}
\end{equation}
When $P_{\mathrm{com}}^{\mathrm{tot}}=0$, the ratio is undefined; however, the above inequality remains applicable and is satisfied since communication leakage is absent.

\subsection{Problem Formulation}\label{sec:problem}
We aim to maximize the communication sum-rate while satisfying the sensing and communication requirements of the proposed OFDM-based ISAR-ISAC system. The constraints include: i) a prescribed MASR threshold to control communication-induced sensing distortion, ii) CU QoS requirements, iii) the required number of active sensing subcarriers, iv) a minimum sensing transmit-power requirement, and v) the total BS transmit-power constraint.

The sensing beamformer $\mathbf{w}$ is designed offline to generate the desired ISAR illumination pattern, e.g., a narrow mainlobe with a prescribed sidelobe level, and is kept fixed during resource allocation. The joint communication-sensing resource-allocation problem is formulated as
\begin{subequations}
\label{P1}
\allowdisplaybreaks
\begin{align}
\max_{\boldsymbol{\rho},\,\boldsymbol{\rho}^{\mathrm{sen}},
\,\boldsymbol{\alpha},\,\mathbf{q}}
\quad
& \sum_{k=1}^{K} R_k
\label{eq:prob_obj_s}\\
\mathrm{s.t.}\quad
&
R_k
\geq
R_k^{\min},
\quad
\forall k,
\label{eq:prob_qos_s}\\
&
\sum_{n=0}^{N_c-1}
\rho_n^{\mathrm{sen}}
=
\widetilde{N}_c,
\label{eq:prob_sensing_cardinality}\\
&
\rho_{k,n}\in\{0,1\},
\quad
\rho_n^{\mathrm{sen}}\in\{0,1\},
\quad
\forall k,n,
\label{eq:prob_binary_s}\\
&
\alpha_{k,n}\geq0,
\quad
\forall k,n,
\\
&
\eqref{eq:total_power_constraint},
\eqref{eq:masr_constr_M},
\label{eq:prob_vars_s}\\
&
q_n\geq \rho_n^{\mathrm{sen}}q_{\min},
\quad \forall n.
\label{eq:prob_sensing_power}
\end{align}
\end{subequations}
Here, $\boldsymbol{\rho}=[\rho_{k,n}]$, $\boldsymbol{\rho}^{\mathrm{sen}}=[\rho_n^{\mathrm{sen}}]$, $\boldsymbol{\alpha}=[\alpha_{k,n}]$, and $\mathbf{q}=[q_0,\ldots,q_{N_c-1}]^T$. The parameter $R_k^{\min}$ denotes the minimum QoS requirement of CU $k$. Constraint \eqref{eq:prob_sensing_cardinality} specifies the number of active sensing subcarriers required for ISAR imaging. Constraint \eqref{eq:prob_sensing_power} guarantees nonzero sensing power on each selected subcarrier, ensuring well-defined observation normalization. We assume $0<q_{\min}<P_{\max}/N_c$ and evaluate the adopted power floor through ISAR reconstruction NMSE simulations. In the full-band sensing mode, $\widetilde{N}_c=N_c$ and $\rho_n^{\mathrm{sen}}=1$ for all $n=0,\ldots,N_c-1$. The optimized sensing-selection vector $\boldsymbol{\rho}^{\mathrm{sen}}$ determines the active sensing-subcarrier set $\Omega_{\mathrm{sen}}= \{n:\rho_n^{\mathrm{sen}}=1\}$. The sampling matrix $\boldsymbol{\Psi}_r$ used in \eqref{eq:2d_sparse_isar_model} is then formed by selecting the rows of the $N_c\times N_c$ identity matrix corresponding to the indices in $\Omega_{\mathrm{sen}}$. It is worth noting that the communication and sensing subcarrier-selection variables are not mutually exclusive. Hence, a subcarrier can be used simultaneously for CU data transmission and sensing waveform transmission. However, activating sensing on subcarrier $n$ introduces the sensing power $q_n$ in the CU SINR denominator, while the communication power on the same subcarrier contributes to the MASR constraint. Therefore, the proposed resource-allocation formulation jointly determines the active sensing subcarriers and the CU allocations according to the communication--sensing tradeoff.

\subsection{Analytical Full-Band MR/ZF Benchmark}
\label{subsec:fullband_benchmark}

To provide an analytical reference for the proposed learning-based resource allocation, we consider the full-band communication and sensing case under both MR and ZF precoding. In this case, all CUs are multiplexed over all subcarriers, and the sensing waveform occupies the whole OFDM bandwidth, i.e., $\rho_{k,n}=1$ and $\rho_n^{\mathrm{sen}}=1$, $\forall k,n$. For analytical tractability, we consider uniform full-band power loading, $\alpha_{k,n}=\bar{\alpha}_k/{N_c}$, $q_n=p_s/N_c$, $\forall k,n$, where $\bar{\alpha}_k$ is the total communication power allocated to CU $k$, and $p_s$ is the total sensing power. Let $S_c=\sum_{k=1}^{K}\bar{\alpha}_k$. Then, using \eqref{eq:sinr_mr_ofdma} and \eqref{eq:sinr_zf_ofdma}, the full-band SINRs for MR and ZF are respectively given by
\begin{equation}
\mathrm{SINR}_{k}^{\mathrm{MR}}
=
\frac{
M_t\eta_k\bar{\alpha}_k
}{
\xi_k(S_c+p_s)+\sigma_k^2
},
\qquad \forall k ,
\label{eq:sinr_fullband_mr}
\end{equation}
and
\begin{equation}
\mathrm{SINR}_{k}^{\mathrm{ZF}}
=
\frac{
(M_t-K)\eta_k\bar{\alpha}_k
}{
\varepsilon_k S_c+\xi_k p_s+\sigma_k^2
},
\qquad \forall k ,
\label{eq:sinr_fullband_zf}
\end{equation}
where $M_t>K$ is required for full-band ZF precoding. Accordingly, for $v\in\{\mathrm{MR},\mathrm{ZF}\}$, the rate of CU $k$ is
\begin{equation}
R_k^{v}
=
\zeta_{\mathrm{dl}}B_{\mathrm{sys}}
\log_2
\left(
1+\mathrm{SINR}_{k}^{v}
\right).
\end{equation}
The corresponding full-band benchmark is formulated as
\begin{subequations}
\label{eq:prob_fullband_benchmark}
\allowdisplaybreaks
\begin{align}
\max_{\{\bar{\alpha}_k\},p_s}
\quad
&
\sum_{k=1}^{K}R_k^{v}
\\
\mathrm{s.t.}\quad
&
R_k^{v}\geq R_k^{\min},
\qquad \forall k,
\\
&
M_t p_0 p_s \geq \kappa S_c,
\label{eq:fullband_masr}
\\
&
p_s\geq N_cq_{\min},
\label{eq:fullband_min_sensing_power}
\\
&
S_c+p_s\leq P_{\max},
\label{eq:fullband_power}
\\
&
\bar{\alpha}_k\geq0,
\qquad \forall k .
\end{align}
\end{subequations}
\begin{Proposition}
\label{prop:fullband_benchmark}
For the full-band benchmark in \eqref{eq:prob_fullband_benchmark}, the optimal sensing power and total communication power are respectively given by
\begin{equation}
p_s^\star
=
\max\left\{
\frac{\kappa P_{\max}}{M_t p_0+\kappa},
N_cq_{\min}
\right\},
\label{eq:ps_star_general}
\end{equation}
and
\begin{equation}
S_c^\star=P_{\max}-p_s^\star.
\label{eq:Sc_star_general}
\end{equation}
Define
\begin{equation}
\chi_k^{\mathrm{MR}}
=
\frac{
M_t\eta_k
}{
\xi_kP_{\max}+\sigma_k^2
},
\label{eq:chi_mr_fullband}
\end{equation}
and, for $M_t>K$,
\begin{equation}
\chi_k^{\mathrm{ZF}}
=
\frac{
(M_t-K)\eta_k
}{
\varepsilon_k S_c^\star+\xi_k p_s^\star+\sigma_k^2
}.
\label{eq:chi_zf_fullband}
\end{equation}
For $v\in\{\mathrm{MR},\mathrm{ZF}\}$, let $\gamma_k = 2^{\frac{R_k^{\min}}{\zeta_{\mathrm{dl}}B_{\mathrm{sys}}}}-1$ and $\bar{\alpha}_{k,v}^{\min} = \frac{\gamma_k}{\chi_k^{v}}$; then the full-band benchmark with precoder $v$ is feasible if and only if
\begin{equation}
\sum_{k=1}^{K}
\bar{\alpha}_{k,v}^{\min}
\leq
S_c^\star .
\label{eq:feasibility_fullband_b}
\end{equation}
If \eqref{eq:feasibility_fullband_b} holds, the optimal communication powers are given by the QoS-aware water-filling solution
\begin{equation}
\bar{\alpha}_{k,v}^{\star}
=
\max
\left\{
\bar{\alpha}_{k,v}^{\min},
\nu_v^\star-\frac{1}{\chi_k^{v}}
\right\},
\qquad \forall k ,
\label{eq:alpha_star_waterfilling_b}
\end{equation}
where $\nu_v^\star$ is selected such that $\sum_{k=1}^{K} \bar{\alpha}_{k,v}^{\star} = S_c^\star$. Finally, the per-subcarrier powers are $\alpha_{k,n}^{\star} = \frac{\bar{\alpha}_{k,v}^{\star}}{N_c}$ and $q_n^\star = \frac{p_s^\star}{N_c}$, $\forall k,n$.
\end{Proposition}

\begin{proof}
Please refer to Appendix~\ref{app:proof_fullband_benchmark}.
\end{proof}
The optimal power allocation follows a truncated water-filling structure, where $\nu_v^\star$ can be efficiently determined via bisection since
\begin{equation}
f_v(\nu)
=
\sum_{k=1}^{K}
\max\left\{
\bar{\alpha}_{k,v}^{\min},
\nu-\frac{1}{\chi_k^v}
\right\}
\end{equation}
is a continuous and nondecreasing function of $\nu$. To avoid a tolerance-dependent numerical search, we use a finite active-set sorting procedure. Specifically, if $\sum_{k=1}^{K}\bar{\alpha}_{k,v}^{\min}=S_c^\star$, then $\bar{\alpha}_{k,v}^{\star}=\bar{\alpha}_{k,v}^{\min}$ for all $k$. Otherwise, define
\begin{equation}
d_{k,v}
=
\bar{\alpha}_{k,v}^{\min}
+
\frac{1}{\chi_k^{v}},
\qquad \forall k .
\label{eq:d_k_waterlevel_b}
\end{equation}
Let $\pi_v(\cdot)$ be a permutation such that
\begin{equation}
d_{\pi_v(1),v}
\leq
d_{\pi_v(2),v}
\leq
\cdots
\leq
d_{\pi_v(K),v}.
\end{equation}
For a candidate active-set size $m$, define
\begin{equation}
\nu_{m,v}
=
\frac{
S_c^\star
-
\sum_{j=m+1}^{K}
\bar{\alpha}_{\pi_v(j),v}^{\min}
+
\sum_{j=1}^{m}
\frac{1}{\chi_{\pi_v(j)}^{v}}
}{m}.
\label{eq:nu_m_waterlevel_b}
\end{equation}
The optimal active-set size $m_v^\star$ is found as an integer satisfying
\begin{equation}
d_{\pi_v(m_v^\star),v}
\leq
\nu_{m_v^\star,v}
\leq
d_{\pi_v(m_v^\star+1),v},
\label{eq:mstar_waterlevel_b}
\end{equation}
where $d_{\pi_v(K+1),v}=\infty$. The optimal water level is then $\nu_v^\star=\nu_{m_v^\star,v}$. The complete procedure is summarized in Algorithm~\ref{alg:fullband_waterfilling}.

\begin{algorithm}[t]
\caption{Full-Band MR/ZF Analytical Benchmark}
\label{alg:fullband_waterfilling}

\KwIn{Precoder $v\in\{\mathrm{MR},\mathrm{ZF}\}$ and system parameters}
\KwOut{$\{\alpha_{k,n}^{\star}\}$ and $\{q_n^\star\}$}

Compute $p_s^\star$, $S_c^\star$ from \eqref{eq:ps_star_general}--\eqref{eq:Sc_star_general}\;

Compute $\chi_k^v$ using \eqref{eq:chi_mr_fullband} or \eqref{eq:chi_zf_fullband}\;

Compute $\gamma_k$, $\bar{\alpha}_{k,v}^{\min}$ for all $k$\;

\If{$\sum_k\bar{\alpha}_{k,v}^{\min}>S_c^\star$}{
Declare infeasible and terminate\;
}

Compute $\nu_v^\star$ using \eqref{eq:d_k_waterlevel_b}--\eqref{eq:mstar_waterlevel_b}\;

Compute $\bar{\alpha}_{k,v}^{\star}$ from \eqref{eq:alpha_star_waterfilling_b}\;

Set $\alpha_{k,n}^{\star}=\bar{\alpha}_{k,v}^{\star}/N_c$, $q_n^\star=p_s^\star/N_c$\;

\end{algorithm}

\section{DRL-Based Solution}
\label{sec:DRL}

In this section, we develop a DRL framework to solve the sum-rate maximization problem \eqref{P1} under sparse-sensing OFDM-ISAC. The joint optimization over communication-subcarrier allocation, sensing-subcarrier selection, communication power control, and dedicated sensing power allocation results in a mixed-integer nonconvex problem that is computationally intractable for real-time implementation.

To address this challenge, we formulate the problem as a Markov decision process (MDP) and develop a solution based on the SAC algorithm. SAC is suitable because it can handle high-dimensional continuous actions and encourages exploration via entropy regularization. The continuous action generated by the actor is mapped into feasible binary subcarrier-selection variables and nonnegative power-allocation variables through deterministic projection operations.

\subsection{MDP Formulation}
\label{subsec:mdp}

The joint communication--sensing resource-allocation problem is formulated as an MDP defined by
\begin{equation}
\mathcal{M}
=
\left\langle
\mathcal{S},
\mathcal{A},
\mathcal{P},
\mathcal{R}
\right\rangle,
\end{equation}
where $\mathcal{S}$ denotes the state space, $\mathcal{A}$ denotes the action space, $\mathcal{P}$ represents the state-transition dynamics, and $\mathcal{R}$ denotes the reward function.

\subsubsection{State Space}

At time step $t$, the state contains the available statistical CSI and noise information. The state is defined as
\begin{equation}
s_t
=
\left[
\boldsymbol{\xi}(t),
\boldsymbol{\eta}(t),
\boldsymbol{\varepsilon}(t),
\boldsymbol{\sigma}(t)
\right],
\end{equation}
where $\boldsymbol{\xi}(t)$ denotes the large-scale fading coefficients,
$\boldsymbol{\eta}(t)$ denotes the channel-estimation gains,
$\boldsymbol{\varepsilon}(t)$ denotes the channel-estimation error variances,
and $\boldsymbol{\sigma}(t)$ denotes the noise-power vector.
Under perfect CSI (PCSI), $\boldsymbol{\varepsilon}(t)=\mathbf 0$,
$\boldsymbol{\eta}(t)=\boldsymbol{\xi}(t)$. The state components are updated across episodes according to the large-scale channel realizations.
\subsubsection{Action Space}
The proposed framework employs the SAC algorithm with a stochastic Gaussian policy. Given the current state $s_t$, the actor network outputs the mean and standard deviation of a multivariate Gaussian distribution,
\begin{equation}
\pi_{\phi}(\mathbf{u}_t|s_t)
=
\mathcal{N}
\!\left(
\boldsymbol{\mu}_{\phi}(s_t),
\operatorname{diag}
\!\left(
\boldsymbol{\sigma}_{\phi}^{2}(s_t)
\right)
\right),
\end{equation}
where $\phi$ denotes the actor-network parameters. Using the reparameterization trick, a latent action vector is generated as
\begin{equation}
\mathbf{u}_t
=
\boldsymbol{\mu}_{\phi}(s_t)
+
\boldsymbol{\sigma}_{\phi}(s_t)
\odot
\boldsymbol{\epsilon}_t,
\qquad
\boldsymbol{\epsilon}_t
\sim
\mathcal{N}(\mathbf{0},\mathbf{I}),
\end{equation}
where $\odot$ denotes element-wise multiplication. To satisfy the bounded action space required by the environment, a component-wise hyperbolic tangent transformation is applied,
\begin{equation}
\mathbf{a}_t
=
\tanh(\mathbf{u}_t),
\end{equation}
which guarantees $a_i(t)\in[-1,1]$, $\forall i$. The resulting action vector is partitioned as
\begin{equation}
\mathbf{a}_t
=
\left[
\mathbf{z}_{\rho}(t),
\mathbf{z}_{\mathrm{sen}}(t),
\mathbf{z}_{\alpha}(t),
\mathbf{z}_{q}(t)
\right],
\end{equation}
where $\mathbf z_{\rho}(t)$ and $\mathbf z_{\mathrm{sen}}(t)$ contain the communication and sensing scheduling scores, respectively, whereas $\mathbf z_{\alpha}(t)$ and $\mathbf z_q(t)$ correspond to the communication and sensing power weights.

\subsubsection{Action Projection}
The continuous scheduling scores are converted into binary scheduling decisions. For communication scheduling,
\begin{equation}
\rho_{k,n}(t)
=
\begin{cases}
1, & z_{\rho,k,n}(t)>0,\\
0, & \text{otherwise}.
\end{cases}
\end{equation}
To avoid degenerate allocations, at least one communication subcarrier is assigned to each user. Whenever $\sum_n\rho_{k,n}(t)=0$, the subcarrier corresponding to the largest scheduling score is activated.
For sensing allocation, the $\widetilde{N}_c$ subcarriers associated with the largest sensing scores are selected. Let $\mathcal{I}_t$ denote the set of indices corresponding to the $\widetilde{N}_c$ largest entries of $\mathbf{z}_{\mathrm{sen}}(t)$. Then,
\begin{equation}
\rho_n^{\mathrm{sen}}(t)
=
\begin{cases}
1, & n\in\mathcal{I}_t,\\
0, & \text{otherwise},
\end{cases}
\end{equation}
which satisfies $\sum_{n=0}^{N_c-1} \rho_n^{\mathrm{sen}}(t)=\widetilde{N}_c$.
The communication- and sensing-power outputs are first transformed into nonnegative weights as
\begin{align}
w_{\alpha,k,n}(t)
&=
\frac{z_{\alpha,k,n}(t)+1}{2}
\,\rho_{k,n}(t),\\
w_{q,n}(t)
&=
\frac{z_{q,n}(t)+1}{2}
\rho_n^{\mathrm{sen}}(t).
\end{align}
To enforce \eqref{eq:prob_sensing_power}, we reserve
$\widetilde N_c q_{\min}$ for the selected sensing subcarriers,
where $\widetilde N_c=|\Omega_{\mathrm{sen}}|$ denotes
the number of selected physical subcarriers.
The residual power is
$P_{\mathrm{res}}=P_{\max}-\widetilde N_c q_{\min}
\geq 0$.
Defining
$W(t)=\sum_{k=1}^{K}\sum_{n=0}^{N_c-1}
w_{\alpha,k,n}(t)
+\sum_{n=0}^{N_c-1}w_{q,n}(t)$,
the communication and sensing powers are obtained
through normalization for $W(t)>0$,
\begin{align}
\alpha_{k,n}(t)
&=
\frac{P_{\mathrm{res}}}{W(t)}
w_{\alpha,k,n}(t),\\
q_n(t)
&=
q_{\min}\rho_n^{\mathrm{sen}}(t)
+\frac{P_{\mathrm{res}}}{W(t)}w_{q,n}(t).
\end{align}
If $W(t)=0$, the residual power is allocated uniformly
among selected sensing subcarriers.
Consequently, the generated action always satisfies
the transmit-power constraint
\eqref{eq:total_power_constraint}
and the sensing-power floor \eqref{eq:prob_sensing_power}.
Furthermore, let
\begin{equation}
A_{\mathrm{s}}(t)
=
\sum_{n=0}^{N_c-1}
\rho_n^{\mathrm{sen}}(t)
\sum_{k=1}^{K}
\rho_{k,n}(t)\alpha_{k,n}(t)
\end{equation}
denote the communication power on the selected sensing
subcarriers.
Whenever the MASR constraint \eqref{eq:masr_constr_M}
is violated, the communication powers are scaled by
\begin{equation}
c(t)=\min\left\{
1,\frac{M_tp_0P_{\max}}
{M_tp_0A_{\mathrm{tot}}(t)+\kappa A_{\mathrm{s}}(t)}
\right\},
\end{equation}
where
$A_{\mathrm{tot}}(t)=
\sum_{k=1}^{K}\sum_{n=0}^{N_c-1}
\rho_{k,n}(t)\alpha_{k,n}(t)$.
The released power is reassigned to the selected sensing
subcarriers without reducing their power floor,
thereby enforcing \eqref{eq:masr_constr_M}.

\subsubsection{Reward Function}
After executing action $a_t$, the environment computes
the achievable user rates $\{R_k(t)\}_{k=1}^{K}$.
For reward evaluation, we use the normalized spectral
efficiencies
$\bar R_k(t)=R_k(t)/(\zeta_{\mathrm{dl}}B_{\mathrm{sys}})$
and
$\bar R_k^{\min}=R_k^{\min}/
(\zeta_{\mathrm{dl}}B_{\mathrm{sys}})$,
in bit/s/Hz.
The immediate reward is defined as
\begin{equation}
\label{reward}
r(t)
=
\sum_{k=1}^{K}\bar R_k(t)
-
\lambda_{\mathrm{qos}}\Delta_{\mathrm{qos}}(t),
\end{equation}
where
$\Delta_{\mathrm{qos}}(t)=
\max_k\!\left(0,\bar R_k^{\min}-\bar R_k(t)\right)$
and $\lambda_{\mathrm{qos}}>0$ is the QoS penalty
coefficient.
The first term in \eqref{reward} encourages high
communication SE, while the penalty
term guides the learned policy toward satisfying
the user QoS constraints.
\subsection{SAC Learning}

The proposed framework employs the SAC algorithm to learn a stochastic policy $\pi_{\phi}(a_t|s_t)$ that maximizes the expected discounted reward while promoting exploration through entropy regularization. The policy is optimized using the objective
\begin{equation}
J_{\pi}
=
\mathbb E_{s_t\sim\mathcal D,\,a_t\sim\pi_{\phi}}
\left[
\alpha_{\mathrm{ent}}
\log \pi_{\phi}(a_t|s_t)
-
Q_{\min}(s_t,a_t)
\right],
\end{equation}
where $\mathcal D$ denotes the replay buffer and
$\alpha_{\mathrm{ent}}$ is the temperature parameter controlling the exploration level. To mitigate overestimation bias, SAC employs two critic networks $Q_{\psi_1}(s,a)$ and $Q_{\psi_2}(s,a)$, and uses
\begin{equation}
Q_{\min}(s,a)
=
\min
\left\{
Q_{\psi_1}(s,a),
Q_{\psi_2}(s,a)
\right\}
\end{equation}
to evaluate the policy objective. The corresponding target critic networks are updated using soft updates.

During training, actions are sampled from the Gaussian policy to encourage exploration. During evaluation, the tanh-transformed mean action $\tanh(\boldsymbol{\mu}_{\phi}(s))$ is employed to obtain deterministic and reproducible resource-allocation decisions. The overall training procedure is summarized in Algorithm~\ref{alg:SAC_ISAC}.

\begin{algorithm}[t]
\caption{SAC-Based Joint Communication--Sensing Resource Allocation}
\label{alg:SAC_ISAC}

\KwIn{$\gamma,\lambda_Q,\lambda_\pi,\lambda_\alpha,\bar{\mathcal H},\tau,B$}

\KwOut{Trained policy $\pi_\phi(a|s)$}

Initialize replay buffer $\mathcal D$, actor $\pi_\phi$, critics
$Q_{\psi_1},Q_{\psi_2}$, target critics
$Q_{\bar\psi_1},Q_{\bar\psi_2}$ with
$\bar\psi_i\leftarrow\psi_i$, and entropy temperature
$\alpha_{\mathrm{ent}}$\;

\For{each episode}{
Generate a channel realization and observe $s_0$\;

\For{each interaction step $t$}{
Sample $a_t\sim\pi_\phi(\cdot|s_t)$ and apply action shaping including sensing-power and MASR projection to obtain
$\{\rho_{k,n}\}$, $\{\rho_n^{\mathrm{sen}}\}$,
$\{\alpha_{k,n}\}$, and $\{q_n\}$\;

Execute the resource-allocation action, obtain reward $r_t$ and next state $s_{t+1}$, and store $(s_t,a_t,r_t,s_{t+1})$ in $\mathcal D$\;

Sample a mini-batch of size $B$ from $\mathcal D$\;

Update critics:
$\psi_i\leftarrow\psi_i-\lambda_Q\nabla_{\psi_i}J_Q(\psi_i)$,
$i=1,2$\;

Update actor:
$\phi\leftarrow\phi-\lambda_\pi\nabla_\phi J_\pi(\phi)$\;

Update entropy temperature:
$\alpha_{\mathrm{ent}}
\leftarrow
\alpha_{\mathrm{ent}}
-\lambda_\alpha\nabla_{\alpha_{\mathrm{ent}}}
J(\alpha_{\mathrm{ent}})$\;

Soft-update target critics:
$\bar\psi_i\leftarrow
\tau\psi_i+(1-\tau)\bar\psi_i$,
$i=1,2$\;
}
}
\end{algorithm}

\section{Simulation Results}
\label{sec:sim}

In this section, we evaluate the performance of the proposed
OFDM-ISAC framework.
We first assess the proposed sparse OFDM-ISAR reconstruction
method, and then investigate joint communication and sensing
resource allocation using the proposed analytical and
learning-based solutions.

\subsection{ISAR Simulation Results}

We evaluate the proposed sparse OFDM-ISAR reconstruction method
using a simulated extended target composed of multiple dominant
scattering centers.
The carrier frequency is $f_c=3.5$ GHz, the system bandwidth
is $B_{\mathrm{sys}}=100$ MHz, the subcarrier spacing is
$\Delta f=60$ kHz, and the number of physical subcarriers is
$N_c=1666$.
The CP duration is $T_g=2.3~\mu$s.
The initial target range is $R_0=5$ km, and the target
angular velocity is $\Omega=0.14$ rad/s.
To emulate sparse sensing, incomplete measurements are
considered in both the slow-time and subcarrier dimensions.
Each coherence block contains $\tau_c=20$ OFDM symbols,
among which $\tau_p=10$ symbols are reserved for pilot
transmission.
Hence, only $50\%$ of the slow-time samples are available
for ISAR sensing.
In addition, $50\%$ of the physical subcarriers are selected
for sensing, resulting in only $25\%$ of the full
slow-time--subcarrier measurement grid.

To link ISAR imaging with resource allocation, the
$N_c=1666$ physical subcarriers are partitioned into
$G=16$ nonoverlapping frequency-distributed groups,
each containing $104$ or $105$ tones.
The SAC-selected groups define the sensing mask
$\Omega_{\mathrm{sen}}$ used for reconstruction.
Throughout the numerical evaluation, the sensing
fraction refers to frequency-domain group selection,
independently of the fixed pilot-induced slow-time
sampling ratio $\widetilde N_s/N_s=0.5$.
At a $50\%$ sensing fraction, eight groups comprise
$833$ tones spanning approximately $99.9$ MHz,
providing broad spectral support for sparse ISAR
reconstruction \cite{Zhang2012SparseISAR,Hashempour2020}.
All reconstruction methods use the same selected
mask and measurements for a fair comparison.
We compare the proposed adaptive 2D-ADMM method with the
zero-filled 2D-FFT and the conventional fixed-$\lambda$
2D-ADMM.
For the 2D-FFT benchmark, the missing samples are filled
with zeros before applying range and Doppler FFTs.
For the ADMM-based methods, reconstruction is performed
directly from the compressed measurements.
Both ADMM methods are implemented using FFT-based forward
and adjoint operators to avoid explicitly forming dense
partial Fourier matrices
and thereby reduce computational complexity.

\begin{figure}[t]
    \centering
    \subfloat[Ground truth\label{fig:isar_gt}]{%
        \includegraphics[width=0.48\linewidth,clip]{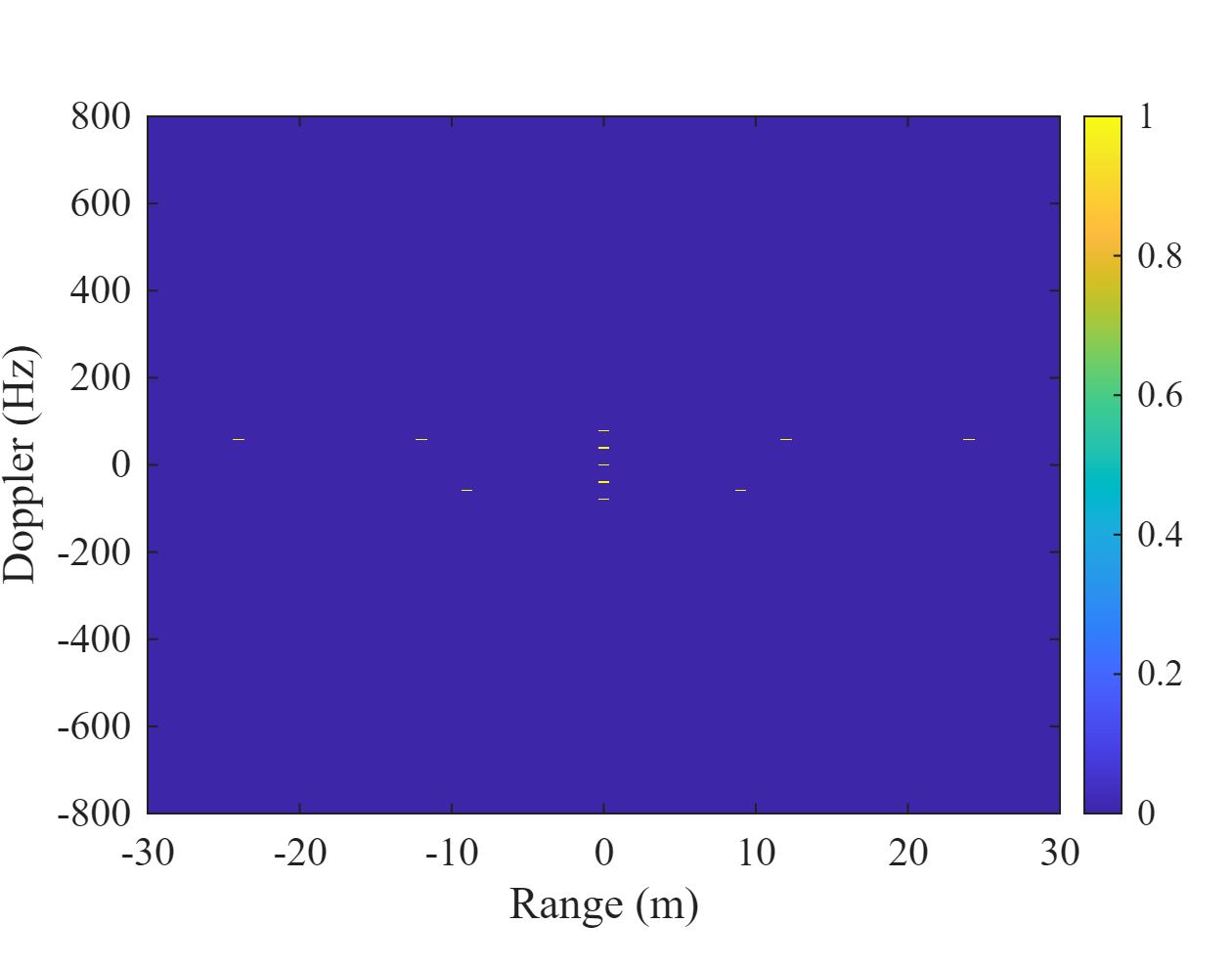}}
    \hfill
    \subfloat[2D-FFT\label{fig:isar_fft_full}]{%
        \includegraphics[width=0.48\linewidth,clip]{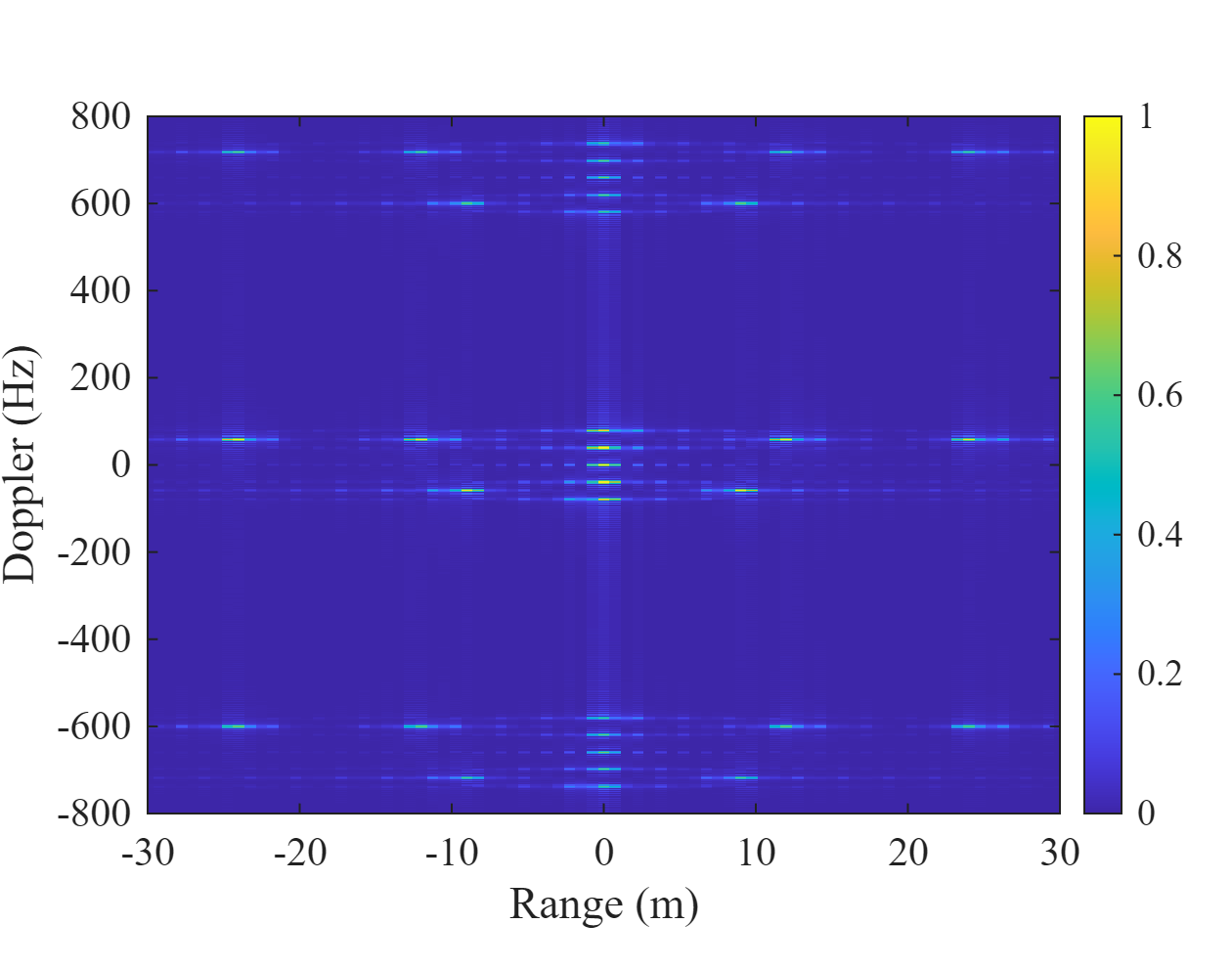}}

    \vspace{-2mm}

    \subfloat[Fixed-$\lambda$ 2D-ADMM\label{fig:isar_admm_full}]{%
        \includegraphics[width=0.48\linewidth,clip]{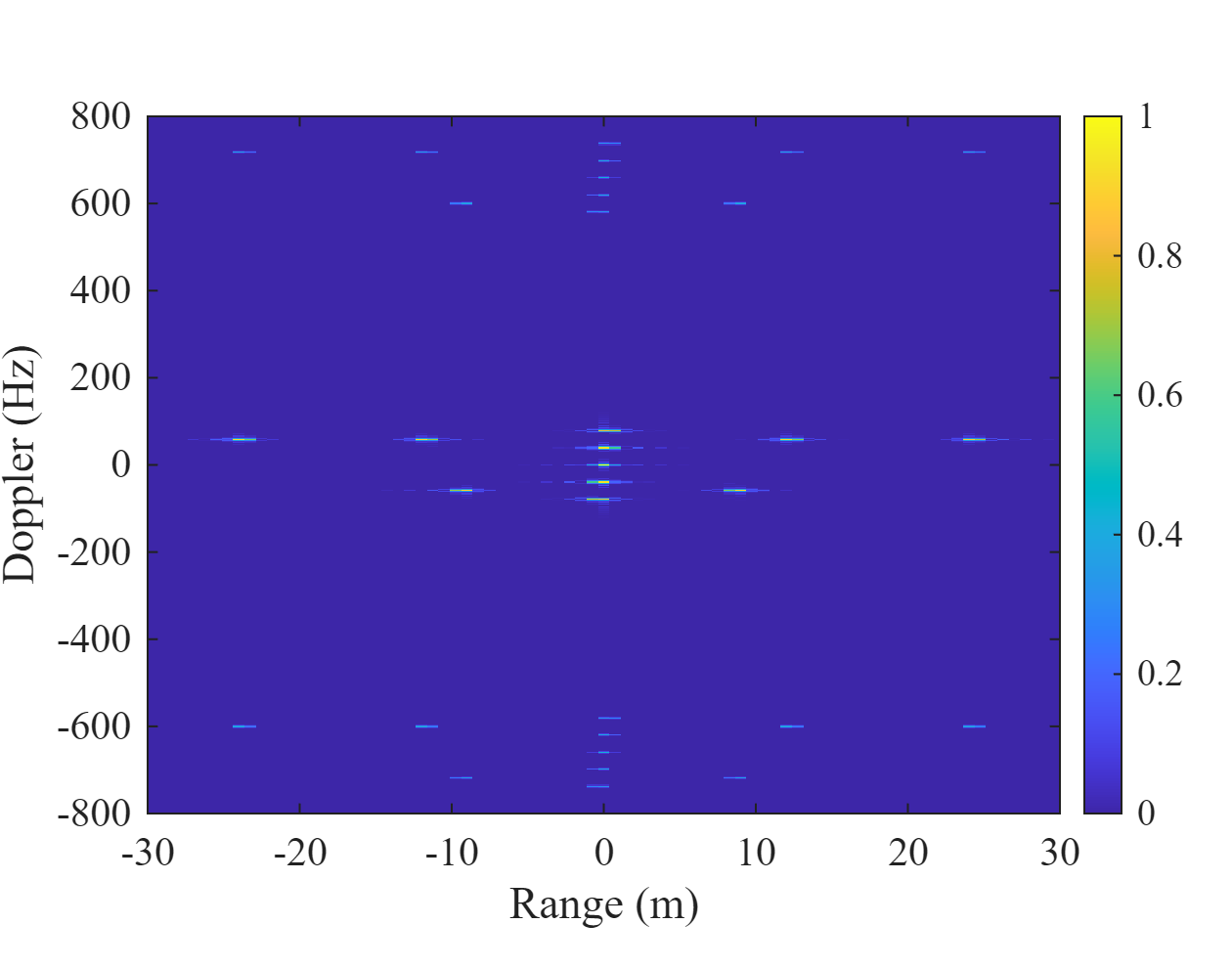}}
    \hfill
    \subfloat[Adaptive 2D-ADMM\label{fig:isar_adaptive_full}]{%
        \includegraphics[width=0.48\linewidth,clip]{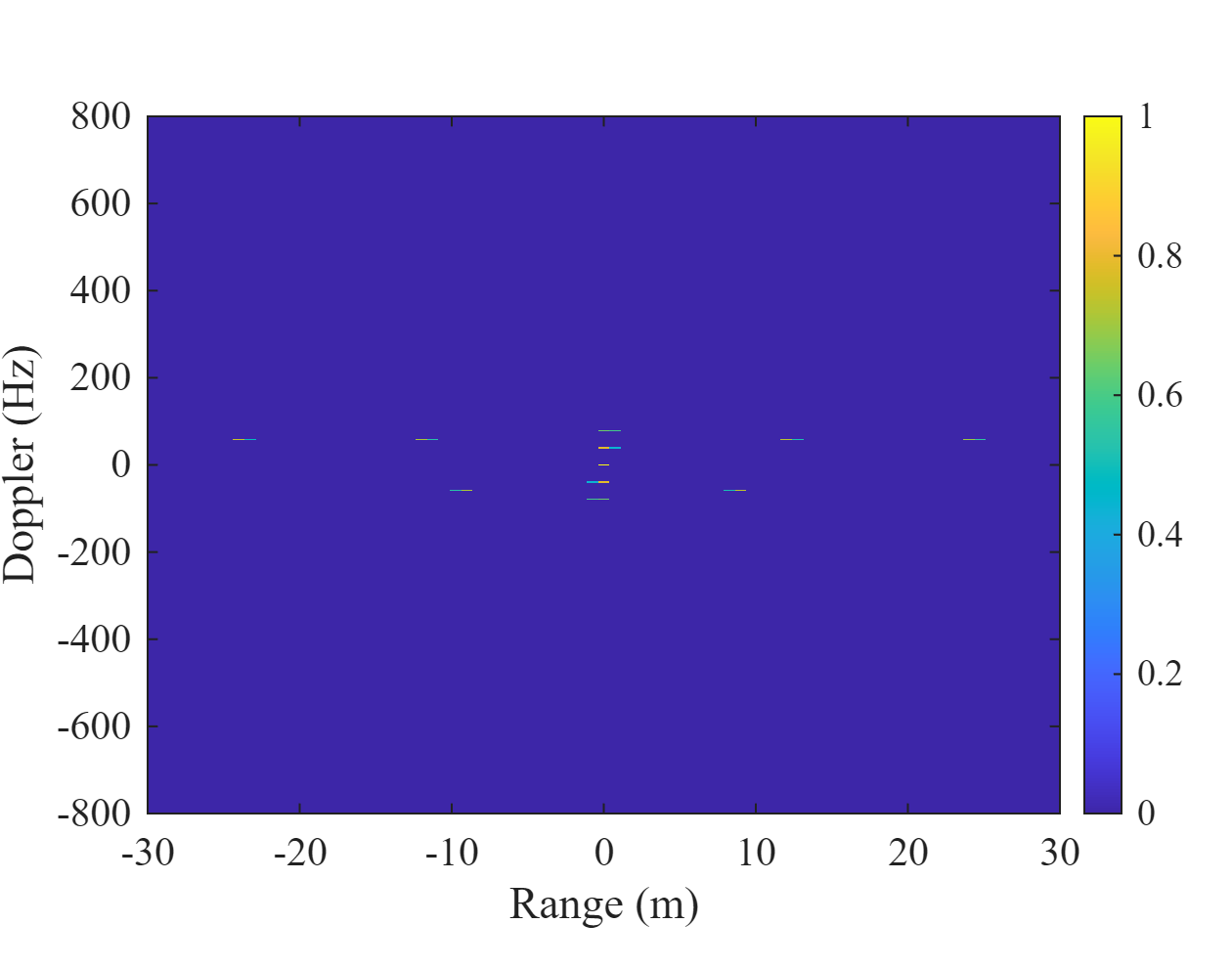}}

   \caption{Sparse ISAR reconstruction using incomplete slow-time and subcarrier measurements at an echo SNR of $20$ dB.}
    \label{fig:isar_sparse_reconstruction}
\end{figure}

Fig.~\ref{fig:isar_sparse_reconstruction} shows the
ground-truth target scene and the reconstructed ISAR images
from sparse measurements at
$20$ dB echo SNR.
The 2D-FFT suffers from strong artifacts and false target
replicas, which are caused by aliasing and sidelobe effects
due to sparse slow-time and subcarrier sampling.
The fixed-$\lambda$ 2D-ADMM suppresses part of these
artifacts by exploiting sparsity, but some residual ghost
scatterers remain.
In contrast, the proposed adaptive 2D-ADMM provides a
cleaner reconstruction, preserving the dominant scattering
centers while effectively suppressing the sampling-induced
artifacts.

Since the ground-truth target image is available, we use
the normalized mean square error (NMSE) as a quantitative
metric.
Let $\mathbf{G}_0$ and $\widehat{\mathbf{G}}$ denote the
ground-truth and reconstructed ISAR images on the same
range--Doppler grid, respectively.
The NMSE is defined as
\begin{equation}
\mathrm{NMSE}
=
10\log_{10}
\left(
\frac{
\left\|
\widehat{\mathbf{G}}-\mathbf{G}_0
\right\|_F^2
}{
\left\|\mathbf{G}_0\right\|_F^2
}
\right).
\label{eq:nmse_metric}
\end{equation}

For quantitative evaluation, the reconstructed images are
represented by their peak-normalized magnitudes.
For Figs.~\ref{fig:isar_sparse_reconstruction}
and \ref{fig:nmse_snr}, sensing power is uniformly
allocated across the selected subcarriers,
subject to the per-tone floor $q_{\min}$.
The echo SNR is adjusted through the receiver-noise variance,
using unit-modulus QPSK sensing symbols.
With residual communication interference neglected,
the normalized noise has equal variance across active
subcarriers, justifying the unweighted reconstruction
objective in \eqref{eq:compact_l1_problem}.

Fig.~\ref{fig:nmse_snr} compares the reconstruction NMSE
versus the sensing echo SNR.
As the SNR increases, all methods exhibit improved reconstruction
accuracy before approaching a high-SNR error floor.
The zero-filled 2D-FFT suffers from sparse-sampling artifacts,
with an NMSE floor of approximately $9.7$ dB, whereas
fixed-$\lambda$ 2D-ADMM reduces this floor to $-28.1$ dB.
The proposed adaptive 2D-ADMM further suppresses reconstruction
artifacts through noise-aware regularization and iterative reweighting.
At $20$ dB SNR, the NMSE reaches $-31.1$ and $-39.4$ dB
for $c_{\mathrm{gap}}=0.15$ and $0.35$, respectively,
with $\eta=0.90$, $\mu=5$, $\tau_{\delta}=2$, $\delta_0=1$,
$\epsilon_{\mathrm{tol}}=10^{-4}$, and $J_{\max}=200$,
achieving up to an $11.3$ dB improvement over fixed-$\lambda$ 2D-ADMM.
\begin{figure}[t]
\centering
\includegraphics[width=0.9\linewidth]{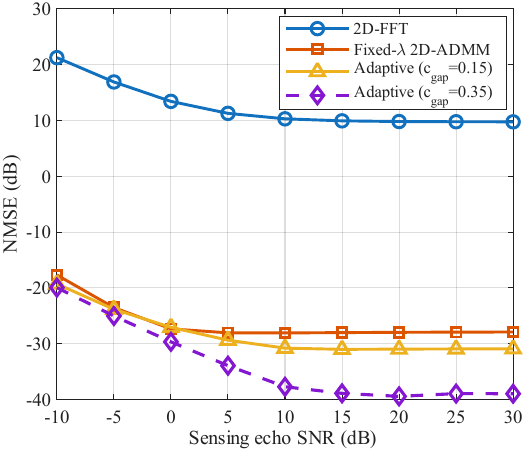}
\caption{NMSE versus sensing echo SNR for sparse OFDM-ISAR reconstruction.}
\label{fig:nmse_snr}
\end{figure}

\subsection{ISAC Simulation Results}

\subsubsection{Simulation Setup}

We consider a single-cell OFDM-ISAC system with a square coverage area of
$500\,\mathrm{m}\times500\,\mathrm{m}$.
The $K$ communication users are uniformly distributed inside the cell while a
minimum BS--user distance $d_{\min}=35$ m is enforced.
The large-scale fading coefficient of user $k$ is modeled as
\cite{Liao2024mMIMOISAC}
\begin{equation}
[\xi_k]_{\mathrm{dB}}
=
-148.1
-
37.6\log_{10}\!\left(\frac{d_k}{1000}\right)
+
\varpi_k,
\end{equation}
where $\varpi_k\sim\mathcal N(0,7^2)$ denotes log-normal shadow fading.
The noise power is set to $-96$ dBm.
Unless otherwise stated, PCSI is assumed.
We evaluate the proposed analytical full-band benchmark
(Algorithm~\ref{alg:fullband_waterfilling}), the SAC-based sparse-sensing
framework (Algorithm~\ref{alg:SAC_ISAC}), and a DDPG-based baseline
under both MR and ZF precoding.
The sensing quality is controlled through the MASR constraint with
threshold $\kappa$, while all users satisfy the same minimum spectral-efficiency requirement $\bar{R}_k^{\min}=R^{\min}$ (i.e., $R_k^{\min}=\zeta_{\mathrm{dl}}B_{\mathrm{sys}}R^{\min}$).

Using the frequency-domain grouping defined in
Section~VI-A, SAC jointly performs communication
scheduling, sensing selection, and power allocation
at the group level, with group powers uniformly
distributed over their constituent physical tones.
Each selected sensing tone satisfies the minimum
power requirement $q_{\min}$.
Sensing fractions of $25\%$ and $50\%$ correspond
to selecting $4$ and $8$ groups, respectively,
whereas $0\%$ sensing serves as a communication-only
reference.
SAC and DDPG are evaluated under identical system
settings and training budgets, with convergence
curves obtained through Monte Carlo averaging
and a $1000$-interaction moving average.
The main system and training parameters are summarized in
Table~\ref{tab:sim_params} unless otherwise specified.

\begin{table}[t]
\centering
\caption{Simulation Parameters}
\label{tab:sim_params}
\begin{tabular*}{\columnwidth}{@{\extracolsep{\fill}}ll@{}}
\hline
\textbf{Parameter} & \textbf{Value} \\
\hline
Cell size & $500\,\text{m}\times500\,\text{m}$ \\
Minimum BS--user distance & $d_{\min}=35\,\text{m}$ \\
Number of antennas & $M_t=400$ \\
Number of users & $K=10$ \\
Number of physical subcarriers & $N_c=1666$ \\
Subcarrier spacing & $\Delta f=60\,\text{kHz}$ \\
Number of resource groups & $G=16$ \\
Transmit power budget & $P_{\max}=30\,\text{dBm}$ \\
Noise power & $-96\,\text{dBm}$ \\
Shadow fading std. & $\sigma_{\mathrm{sf}}=7\,\text{dB}$ \\
Pathloss constant & $-148.1\,\text{dB}$ \\
Pathloss slope & $37.6$ \\
Minimum SE requirement & $R^{\min}=0.5\,\text{bit/s/Hz}$  \\
MASR threshold & $\kappa=30\,\text{dB}$  \\
CSI mode & PCSI \\
Precoding schemes & MR, ZF \\
Sensing power floor & $q_{\min}=10^{-6}\,\text{W/tone}$ \\
\hline
\multicolumn{2}{c}{\textit{SAC Training Parameters}} \\
\hline
Batch size & $256$ \\
Replay buffer size & $3\times10^5$ \\
Discount factor & $\gamma=0.99$ \\
Soft update factor & $\tau=5\times10^{-3}$ \\
Actor learning rate & $10^{-4}$ \\
Critic learning rate & $10^{-4}$ \\
Entropy learning rate & $3\times10^{-4}$ \\
Training steps & $2\times10^4$ \\
QoS penalty coefficient & $\lambda_{\mathrm{qos}}=200$ \\
\hline
\end{tabular*}
\end{table}

\subsubsection{Performance Evaluation}
For performance comparison consistent with the reward in \eqref{reward}, we report the downlink-phase sum SE (SSE), defined as $\mathrm{SSE}\triangleq\sum_{k=1}^{K}\bar R_k$, in bit/s/Hz, with minimum threshold $\bar R_k^{\min}=R^{\min}$.

Fig.~\ref{fig:wf_benchmark} presents the average SSE of the
analytical full-band benchmark under MR and ZF precoding,
evaluated over $2000$ independent channel realizations.
Increasing the transmit-power budget improves the SSE,
whereas a higher MASR threshold allocates more power
to sensing, reducing communication performance.
ZF outperforms MR due to its superior interference suppression,
although the performance gap narrows as the sensing-power
requirement becomes more stringent.

\begin{figure}[t]
\centering
\includegraphics[width=0.9\linewidth]{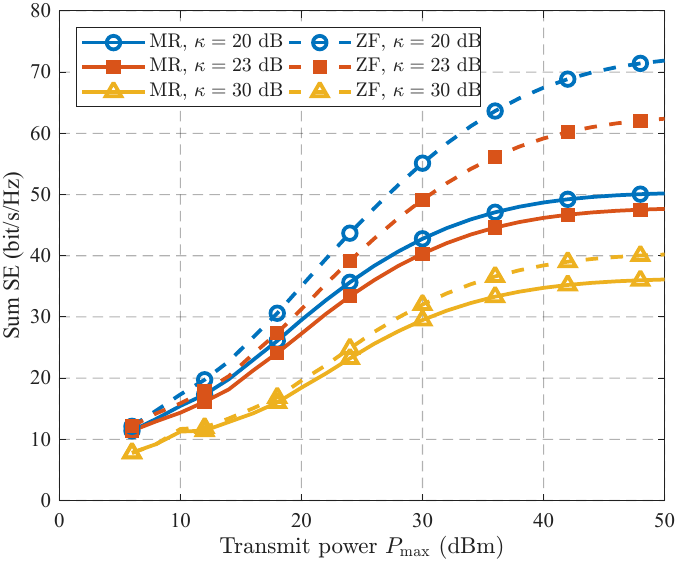}
\caption{SSE versus transmit power for the analytical full-band benchmark.}
\label{fig:wf_benchmark}
\end{figure}

Fig.~\ref{fig:icsi_error} illustrates the impact of normalized
channel estimation error on the SSE
at $P_{\max}=40$ dBm
for different sensing and QoS requirements.
The SSE gradually decreases with increasing channel estimation error.
For $\kappa=25$ dB, both precoders are relatively robust to CSI
imperfections, and ZF provides a clear gain over MR.
The curves for $R^{\min}=0.5$ and $4$ bit/s/Hz/user
nearly coincide over the common feasible channel
realizations, indicating that the stricter QoS requirement
has negligible impact on the optimal power allocation.
When the sensing requirement increases to $\kappa=35$ dB,
the communication power budget is further reduced,
resulting in a noticeable performance loss.
In particular, the case $\kappa=35$ dB and $R^{\min}=4$
is infeasible for both MR and ZF
and is therefore omitted from the figure.
\begin{figure}[t]
\centering
\includegraphics[width=0.9\linewidth]{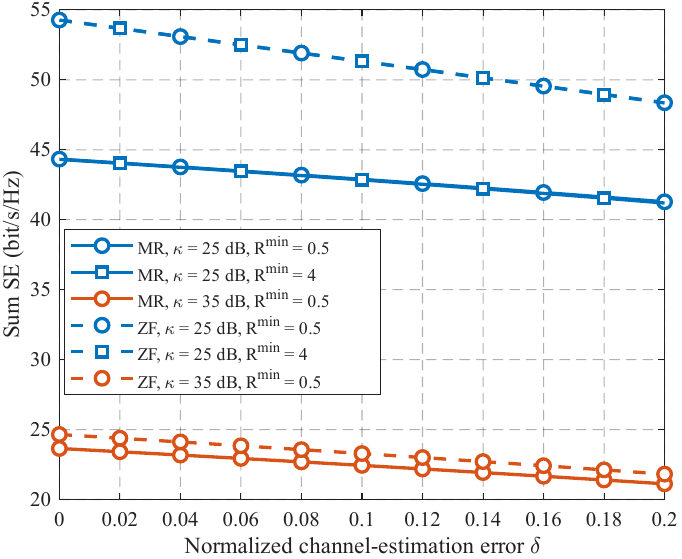}
\caption{SSE versus normalized channel estimation error $\delta$
for MR and ZF precoding under different sensing thresholds
$\kappa$ and $R^{\min}$.}
\label{fig:icsi_error}
\end{figure}

Figs.~\ref{fig:sac_mr_convergence} and
\ref{fig:sac_zf_convergence} illustrate the training convergence
of the SAC-based resource-allocation framework under MR and ZF
precoding, respectively.
Following the initial $2000$-step exploration stage, the sum SE
increases rapidly and stabilizes within $2\times10^4$ steps.
For both precoders, the $0\%$ sensing case closely approaches
the analytical full-band communication-only benchmark
($44.5$ vs.\ $45.1$ bit/s/Hz under MR, and
$72.5$ vs.\ $73.5$ bit/s/Hz under ZF).
At $25\%$ sensing, SAC achieves $40.2$ and $62.4$ bit/s/Hz
under MR and ZF precoding, respectively, outperforming
the corresponding DDPG baseline ($30.6$ and $52.8$ bit/s/Hz)
and analytical full-band ISAC benchmark
($28.9$ and $31.3$ bit/s/Hz).
Reducing the sensing fraction from $50\%$ to $25\%$
provides greater flexibility in communication scheduling
and power allocation, thereby improving the sum SE.
Moreover, the final evaluated policies achieve an empirical
all-user QoS-feasibility rate of $100\%$ under the
considered simulation settings.

\begin{figure}[t]
\centering
\includegraphics[width=0.9\linewidth]{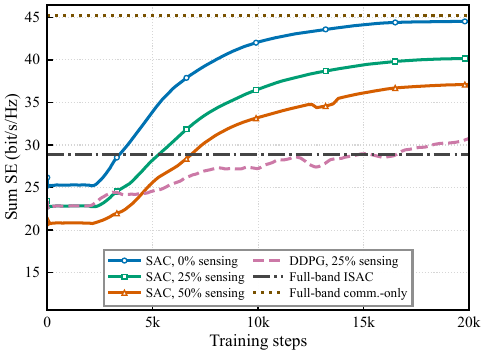}
\caption{Training convergence of SAC with different
sensing fractions under MR precoding, compared with DDPG and
analytical full-band benchmarks.}
\label{fig:sac_mr_convergence}
\end{figure}

\begin{figure}[t]
\centering
\includegraphics[width=0.9\linewidth]{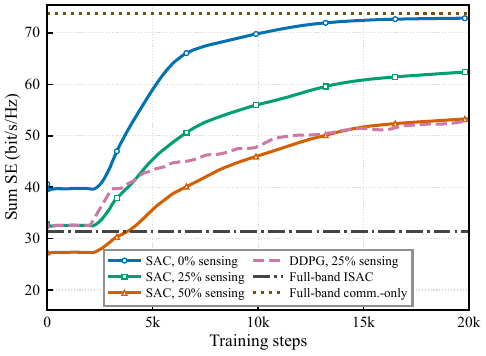}
\caption{Training convergence of SAC with different
sensing fractions under ZF precoding, compared with DDPG and
analytical full-band benchmarks.}
\label{fig:sac_zf_convergence}
\end{figure}

\begin{figure}[t]
\centering
\includegraphics[width=0.9\linewidth]{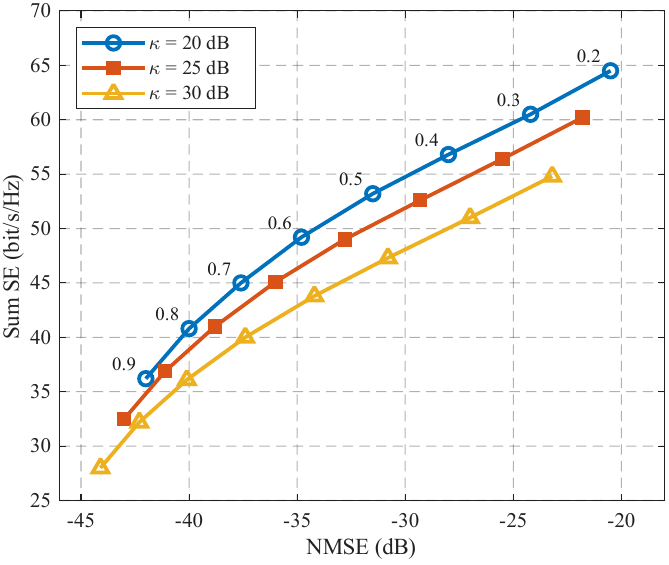}
\caption{Communication sum SE versus ISAR reconstruction
NMSE under ZF precoding for different sensing fractions and MASR thresholds $\kappa$.}
\label{fig:SSE_vs_NMSE}
\end{figure}

Fig.~\ref{fig:SSE_vs_NMSE} illustrates the tradeoff between
communication sum SE and ISAR reconstruction NMSE under ZF precoding
($c_{\mathrm{gap}}=0.35$, nominal echo SNR of $20$ dB) using SAC-selected
sensing masks with nominal sensing fractions ranging from
$0.2$ to $0.9$ and $\kappa\in\{20,25,30\}$ dB.
The actual sensing fractions are determined by the selected
resource groups.
Increasing the sensing fraction provides additional frequency-domain
observations, improving the reconstruction NMSE from $-20.5$ to
$-42.0$ dB for $\kappa=20$ dB, while reducing communication
sum SE from $64.5$ to $36.2$ bit/s/Hz.
Similarly, increasing $\kappa$ shifts more power toward sensing
at the expense of communication performance.
These results are specific to the considered simulation setting,
as practical MASR requirements depend on site-specific clutter
and interference conditions.
\section{Conclusion}\label{sec:conc}

This paper developed a unified mMIMO-OFDM ISAC framework
for sparse ISAR imaging and joint communication-sensing
resource allocation. An adaptive reweighted 2D-ADMM
algorithm was proposed for ISAR reconstruction from
incomplete measurements, while analytical full-band
benchmarks were derived for MR and ZF precoding using
statistical CSI. A SAC-based framework was further developed
to jointly optimize subcarrier assignment and transmit power
under communication and sensing constraints.
Numerical results demonstrated that the adaptive ADMM
achieves up to an $11.3$ dB NMSE improvement over
conventional fixed-$\lambda$ ADMM at $20$ dB echo SNR.
Moreover, SAC outperforms both analytical full-band ISAC
benchmarks and DDPG in sum SE under MR and ZF precoding.
The results also characterize the tradeoff between
communication SE and ISAR reconstruction accuracy.
Future work will investigate imaging-aware resource allocation
with explicit reconstruction-quality constraints and
robustness to practical clutter and residual interference.

\appendices
\numberwithin{equation}{section}

\section{Proof of Proposition~\ref{prop:mr_uatf}}
\label{app:proof_mr_uatf}
We fix an arbitrary subcarrier $n$ and omit the index $n$ in all subcarrier-dependent variables for notational simplicity. Using \eqref{eq:x_ln} and \eqref{eq:mr_precoder_def}, the received signal in \eqref{eq:cu_rx_signal} can be rewritten as
\begin{align}
y_{k,\ell}
&=
\mathbf{g}_k^H
\left(
\sum_{i=1}^{K}
\rho_i\mathbf{t}_{\mathrm{c},i}s_i[\ell]
+
\sqrt{q\rho^{\mathrm{sen}}}\,
\mathbf{w}s^{\mathrm{sen}}[\ell]
\right)
+
v_{k,\ell}
\nonumber\\
&=
A_k s_k[\ell]
+
\sum_{i\neq k}
B_{k,i}s_i[\ell]
+
C_k s^{\mathrm{sen}}[\ell]
+
v_{k,\ell},
\label{eq:mr_y_decomp}
\end{align}
where the effective coefficients are defined as $A_k \triangleq \rho_k\mathbf{g}_k^H\mathbf{t}_{\mathrm{c},k}$, $B_{k,i} \triangleq \rho_i\mathbf{g}_k^H\mathbf{t}_{\mathrm{c},i}\ (i\neq k)$ and $C_k \triangleq \sqrt{q\rho^{\mathrm{sen}}}\, \mathbf{g}_k^H\mathbf{w}$. Applying the UatF bound \cite{Marzetta2016}, the SINR is
\begin{align}
   & \mathrm{SINR}^{\mathrm{eff}}
=\nonumber\\&
\frac{|\mathbb{E}[A_k]|^2}
{\mathbb{E}[|A_k|^2]-|\mathbb{E}[A_k]|^2
+\sum_{i\neq k}\mathbb{E}[|B_{k,i}|^2]
+\mathbb{E}[|C_k|^2]
+\sigma_k^2 }.
\label{eq:uatf_general}
\end{align}
Next, by invoking $\|\mathbf{w}\|_2^2=1$, the channel decomposition \eqref{eq:g_hat}--\eqref{eq:eps}, and the MR precoder in \eqref{eq:mr_precoder_def}, we obtain
\begin{align}
\mathbb{E}[A_k]
&=\rho_k\sqrt{\frac{\alpha_k}{M_t}}\;\mathbb{E}\{\mathbf{g}_k^H\bar{\mathbf{g}}_k\}=\rho_k\sqrt{\frac{\alpha_k}{M_t}}\,
\mathbb{E}\{\widehat{\mathbf{g}}_k^H\bar{\mathbf{g}}_k\}\nonumber\\&
=\rho_k\sqrt{M_t\eta_k\alpha_k},
\label{eq:mr_EA}
\end{align}
where we used $\mathbb{E}[\|\bar{\mathbf{g}}_k\|_2^2]=M_t$ and $\mathbb{E}[\mathbf{e}_k^H\bar{\mathbf{g}}_k]=0$. Moreover,
\begin{align}
\allowdisplaybreaks
\mathbb{E}[|A_k|^2]
&=\frac{\rho_k\alpha_k}{M_t}\,\mathbb{E}\!\left[|\mathbf{g}_k^H\bar{\mathbf{g}}_k|^2\right] \nonumber \\&
=\frac{\rho_k\alpha_k}{M_t}\Big(\mathbb{E}[|\widehat{\mathbf{g}}_k^H\bar{\mathbf{g}}_k|^2]+\mathbb{E}[|\mathbf{e}_k^H\bar{\mathbf{g}}_k|^2]\Big)\nonumber\\
&=\frac{\rho_k\alpha_k}{M_t}\Big(\eta_k\,\mathbb{E}[\|\bar{\mathbf{g}}_k\|_2^4]+\varepsilon_k\,\mathbb{E}[\|\bar{\mathbf{g}}_k\|_2^2]\Big)
\nonumber \\&=(M_t+1)\eta_k \rho_k\alpha_k+\varepsilon_k \rho_k\alpha_k
\nonumber \\&=M_t\eta_k \rho_k\alpha_k+\xi_k \rho_k\alpha_k,
\label{eq:mr_EA2}
\end{align}
where we used $\mathbb{E}[\|\bar{\mathbf{g}}_k\|_2^4]=M_t(M_t+1)$ and $\xi_k=\eta_k+\varepsilon_k$. For $i\neq k$, since $\bar{\mathbf{g}}_i$ is independent of $\mathbf{g}_k$,
\begin{align}
\mathbb{E}[|B_{k,i}|^2]
&=\frac{\rho_i\alpha_i}{M_t}\,\mathbb{E}\!\left[|\mathbf{g}_k^H\bar{\mathbf{g}}_i|^2\right]
=\frac{\rho_i\alpha_i}{M_t}\,\mathbb{E}\!\left[\|\mathbf{g}_k\|_2^2\right]
=\xi_k\,\rho_i\alpha_i,
\label{eq:mr_EB2}
\end{align}
where $\mathbb{E}[\|\mathbf{g}_k\|_2^2]=M_t\xi_k$. Finally, for the dedicated sensing term,
\begin{equation}
\mathbb{E}[|C_k|^2]
=q\rho^{\mathrm{sen}}\,\mathbb{E}\!\left[|\mathbf{g}_k^H\mathbf{w}|^2\right]
=q\rho^{\mathrm{sen}}\,\xi_k\,\|\mathbf{w}\|_2^2
=q\rho^{\mathrm{sen}}\,\xi_k.
\label{eq:mr_EC2}
\end{equation}
Substituting \eqref{eq:mr_EA}--\eqref{eq:mr_EC2} into \eqref{eq:uatf_general} yields
\begin{equation}
\mathrm{SINR}^{\mathrm{eff,MR}}_{k}
=
\frac{
M_t\eta_k\rho_k\alpha_k
}{
\xi_k\sum_{i=1}^{K}\rho_i\alpha_i
+
\xi_k\rho^{\mathrm{sen}}q
+
\sigma_k^2
}.
\end{equation}
Restoring the subcarrier index $n$ gives \eqref{eq:sinr_mr_ofdma}.

\section{Proof of Proposition~\ref{prop:zf_uatf}}
\label{app:proof_zf_uatf}

For notational simplicity, we omit the subcarrier index $n$ in this proof. Under the ZF precoder in \eqref{eq:zf_precoder_def}, the beamforming vectors are constructed such that \(\widehat{\mathbf{g}}_k^H\mathbf{t}_{\mathrm{c},i}=0\) for all $i\in\mathcal{U}$ with $i\neq k$. Hence, the multiuser interference originating from the estimated channels is completely eliminated, and the residual interference arises solely from the channel estimation error $\mathbf{e}_k$. Proceeding as in \eqref{eq:mr_y_decomp}--\eqref{eq:uatf_general}, define
\begin{align}
\mu_k &\triangleq \mathbb{E}\!\left[\mathbf{g}_k^H\mathbf{t}_{\mathrm{c},k}\right],\\
\eta_{k,i} &\triangleq \mathbb{E}\!\left[\left|\mathbf{g}_k^H\mathbf{t}_{\mathrm{c},i}\right|^2\right],\\
\eta_{k}^{\mathrm{w}} &\triangleq \mathbb{E}\!\left[\left|\mathbf{g}_k^H\mathbf{w}\right|^2\right].
\end{align}
Then, the UatF effective SINR for CU $k$ is
\begin{equation}
\mathrm{SINR}_k^{\mathrm{eff}}
=
\frac{\rho_k|\mu_k|^2}
{\eta_{k,k}-|\mu_k|^2
+\sum_{i\in\mathcal{U},\,i\neq k}\rho_i\eta_{k,i}
+q\,\rho^{\mathrm{sen}}\eta_{k}^{\mathrm{w}}
+\sigma_k^2 }.
\end{equation}
Using standard properties of ZF precoding with i.i.d.\ Rayleigh fading channel estimates \cite{Liao2024mMIMOISAC}, we obtain
\begin{equation}
\mathbb{E}[\mathbf{g}_k^H\mathbf{t}_{\mathrm{c},k}]
=
\sqrt{(M_t-L)\,\eta_k\,\alpha_k},
\end{equation}
and
\begin{equation}
\sum_{i\in\mathcal{U}}
\mathbb{E}[|\mathbf{e}_k^H\mathbf{t}_{\mathrm{c},i}|^2]
=
\varepsilon_k\sum_{i\in\mathcal{U}}\alpha_i,
\end{equation}
where $L\triangleq|\mathcal{U}|$ denotes the number of scheduled users on the considered subcarrier. Moreover, since $\|\mathbf{w}\|_2^2=1$, we have
\begin{equation}
\mathbb{E}
\left[
|\mathbf{g}_k^H\mathbf{w}|^2
\right]
=
\xi_k.
\end{equation}
Substituting the above expectations into the UatF SINR expression and restoring the subcarrier index $n$ yields \eqref{eq:sinr_zf_ofdma}, which completes the proof.

\section{Proof of Proposition~\ref{prop:fullband_benchmark}}
\label{app:proof_fullband_benchmark}

Suppose that neither the MASR constraint \eqref{eq:fullband_masr} nor the minimum sensing-power constraint \eqref{eq:fullband_min_sensing_power} is active. Then, for a fixed total transmit power $S_c+p_s$, one can reduce $p_s$ and increase one of the communication powers $\bar{\alpha}_k$ by the same sufficiently small amount while preserving feasibility. For MR precoding, the SINR denominator in \eqref{eq:sinr_fullband_mr} remains unchanged. For ZF precoding, the denominator in \eqref{eq:sinr_fullband_zf} does not increase, since $\xi_k=\eta_k+\varepsilon_k$ and $\varepsilon_k\leq\xi_k$. Meanwhile, the numerator of the selected CU increases, thereby improving the sum rate, which contradicts optimality. Therefore, at least one of the two sensing constraints must be active at the optimum.

The total-power constraint \eqref{eq:fullband_power} must also be active. Otherwise, all communication and sensing powers could be scaled up by a common factor while preserving the MASR constraint. Such scaling also preserves the minimum sensing-power constraint. Since the receiver noise powers are positive, i.e., $\sigma_k^2>0$, this scaling increases the SINRs and hence increases the sum rate. Therefore,
\begin{equation}
S_c^\star+p_s^\star
=
P_{\max}.
\label{eq:active_power_appendix}
\end{equation}
Using \eqref{eq:active_power_appendix}, the MASR constraint yields
\begin{equation}
p_s\geq
\frac{\kappa P_{\max}}
{M_t p_0+\kappa}.
\end{equation}
Together with $p_s\geq N_cq_{\min}$, and since at least one sensing constraint is active at the optimum, we obtain \eqref{eq:ps_star_general}. The optimal communication power then follows from \eqref{eq:Sc_star_general}.

Substituting \eqref{eq:active_power_appendix} into \eqref{eq:sinr_fullband_mr} gives
\begin{equation}
\mathrm{SINR}_{k}^{\mathrm{MR}}
=
\chi_k^{\mathrm{MR}}\bar{\alpha}_k ,
\end{equation}
where $\chi_k^{\mathrm{MR}}$ is defined in \eqref{eq:chi_mr_fullband}. Similarly, substituting \eqref{eq:ps_star_general} and \eqref{eq:Sc_star_general} into \eqref{eq:sinr_fullband_zf} gives
\begin{equation}
\mathrm{SINR}_{k}^{\mathrm{ZF}}
=
\chi_k^{\mathrm{ZF}}\bar{\alpha}_k ,
\end{equation}
where $\chi_k^{\mathrm{ZF}}$ is defined in \eqref{eq:chi_zf_fullband}. Hence, for $v\in\{\mathrm{MR},\mathrm{ZF}\}$, the QoS constraint $R_k^v\geq R_k^{\min}$ is equivalent to
\begin{equation}
\bar{\alpha}_k
\geq
\bar{\alpha}_{k,v}^{\min}
=
\frac{\gamma_k}{\chi_k^v},
\end{equation}
where $\gamma_k=2^{R_k^{\min}/(\zeta_{\mathrm{dl}}B_{\mathrm{sys}})}-1$. Therefore, after fixing the optimal communication--sensing power split, and dropping the positive constant $\zeta_{\mathrm{dl}}B_{\mathrm{sys}}$, the remaining communication power allocation problem for precoder $v$ becomes
\begin{subequations}
\label{eq:reduced_power_allocation_b}
\begin{align}
\max_{\{\bar{\alpha}_k\}}
\quad
&
\sum_{k=1}^{K}
\log_2(1+\chi_k^v\bar{\alpha}_k)
\\
\mathrm{s.t.}\quad
&
\sum_{k=1}^{K}\bar{\alpha}_k=S_c^\star,
\\
&
\bar{\alpha}_k\geq \bar{\alpha}_{k,v}^{\min},
\qquad \forall k .
\end{align}
\end{subequations}
The problem in \eqref{eq:reduced_power_allocation_b} is a concave maximization problem with linear constraints. Thus, the KKT conditions are necessary and sufficient for global optimality~\cite{boyd2004convex}. Applying the standard water-filling argument~\cite{tse2005fundamentals}, the optimal communication powers are given by
\begin{equation}
\bar{\alpha}_{k,v}^{\star}
=
\max
\left\{
\bar{\alpha}_{k,v}^{\min},
\nu_v^\star-\frac{1}{\chi_k^v}
\right\},
\qquad \forall k ,
\end{equation}
where $\nu_v^\star$ is selected such that
\begin{equation}
\sum_{k=1}^{K}
\bar{\alpha}_{k,v}^{\star}
=
S_c^\star .
\end{equation}
Finally, the feasibility condition follows directly from the requirement that the available communication power $S_c^\star$ must be no smaller than the sum of the minimum QoS powers, i.e.,
\begin{equation}
\sum_{k=1}^{K}
\bar{\alpha}_{k,v}^{\min}
\leq
S_c^\star .
\end{equation}
This completes the proof.

\bibliographystyle{ieeetr}
\bibliography{references}

\end{document}